 \newcommand{\mydriver}{pdflatex} 

\documentclass[12pt,\mydriver]{thesis}  

\usepackage{amsfonts}
\usepackage{amsmath}
\usepackage{cool}
\usepackage{amsthm}
\usepackage{ragged2e}

\usepackage{mathtools}
\usepackage{float}
\usepackage{graphicx}
\usepackage{dsfont}

\newtheorem{lemma}{Lemma}
\newtheorem{theorem}{Theorem}
\newtheorem{proposition}{Proposition}
\theoremstyle{definition}
\newtheorem{definition}{Definition}
\theoremstyle{remark}
\newtheorem{remark}{Remark}

\theoremstyle{definition}
\newtheorem{example}{Example}

\usepackage{titlesec}
   \titleformat{\chapter}
      {\normalfont\large}{Chapter \thechapter:}{1em}{}
\usepackage{lmodern}
\usepackage{babel}
\usepackage{graphicx}
\usepackage{lscape}
\usepackage{indentfirst}
\usepackage{latexsym}
\usepackage{multirow}
\usepackage{epstopdf}
\usepackage{tabls}
\usepackage{wrapfig}
\usepackage{slashbox}

\usepackage{longtable}
\usepackage{supertabular}
\usepackage[table]{xcolor}
\usepackage{setspace}
\usepackage[sort,square,numbers]{natbib}
\usepackage[colorlinks=true,urlcolor=black,linkcolor=blue,citecolor=blue]{hyperref}
\usepackage{theoremref}
\usepackage{tocloft,calc}
\renewcommand{\cftchapaftersnum}{:\ }
\renewcommand{\cftchappresnum}{\chaptername\space}
\makeatletter
\g@addto@macro\appendix{%
  \addtocontents{toc}{%
    \protect%
  }%
} 

\renewcommand{\baselinestretch}{2}
\begin{document}


\pagestyle{empty}

\thispagestyle{empty}
\hbox{\ }
\vspace{1in}
\renewcommand{\baselinestretch}{1}
\small\normalsize
\begin{center}

\textbf{{Controllability and Tracking of Ensembles: \\
An Optimal Transport Theory Viewpoint}}\\

\ \\
\ \\
 Reza Hadadi
\\ Department of Electrical and Computer Engineering\\
University of Maryland\\
College Park, MD 20784 \\
\texttt{rezahd@umd.edu} 
\vspace{0.4in}

October 2021
\end{center}
\def\abstract{
  \vspace{0.5in}
\begin{center}%
{\bfseries \abstractname\vspace{-.5em}}%
\end{center}
\quotation
}

\def\endabstract{\par \endquotation
}
\begin{abstract}

    \noindent This paper explores the controllability and state tracking of ensembles from the perspective of optimal transport theory. Ensembles, characterized as collections of systems evolving under the same dynamics but with varying initial conditions, are a fundamental concept in control theory and applications. By leveraging optimal transport, we provide a novel framework for analyzing and solving the state tracking problem of ensembles, particularly when state observations are limited and only accessible at discrete time points. This study establishes connections between the ensemble dynamics and finite-horizon optimal control problems, demonstrating that the problem can be reformulated as a com\-pu\-ta\-tion\-al\-ly efficient linear program using Kantorovich's formulation of optimal transport. We raise notions of observability and controllability for nonlinear ensembles, and propose methods for state tracking in Gaussian output dis\-tri\-bu\-tions settings. Numerical examples and theoretical insights are provided to validate the approach, highlighting the utility of optimal transport in ensemble control problems.
    \end{abstract}


\pagestyle{plain} \pagenumbering{roman} \setcounter{page}{2}
    \cleardoublepage
    \phantomsection
    \addcontentsline{toc}{chapter}{Table of Contents}
    \renewcommand{\contentsname}{Table of Contents}
\renewcommand{\baselinestretch}{1}
\small\normalsize
\tableofcontents 
\newpage
\newpage
\newpage
\newpage
\setlength{\parskip}{0em}
\renewcommand{\baselinestretch}{2}
\small\normalsize

\setcounter{page}{1}
\pagenumbering{arabic}

\renewcommand{\thechapter}{1}

\chapter{Introduction}

There are many circumstances in engineering and sciences in which a great (finite or infinite) number of particles, agents, vehicles, animals, systems etc., are interacting or evolving with each other over time in a dynamic fashion such that the evolution of their so called \textit{states} in time may be characterized by a set of differential equations, or more generally, by dynamical systems in its mathematical sense. Ideally, we as engineers are interested in influencing such interactions in a way that the overall collective behavior of the agents possess some specified characteristics of desire. Furthermore, there are cases in which it is of importance not to influence the agents states  necessarily, but to track or specify states of the agents over time based on some set of observations which are accessible to us only at certain time instances, and which are a priori known to be affected by the agents states. In this paper, we are aiming at studying a major class of collections of agents, called ensembles.\\

Ensembles are considered that are collections of systems in finite or infinite number, where all systems evolve under the same dynamic equations with (possibly) different initial states, while they are not distinguishable individually. Instead, some observations, as functions, of the joint states of the systems are available. On top of that, the limitations imposed by the sensing instruments brings about observations which are available only at some time moments, say $t=0,1,2,...,T$. In this setting our goal is to interpolate, a.k.a track, the joint states or their distributions at any $t \in [0, T]$, by using optimal transport theory. Optimal transport theory has to do with analysing and obtaining the map(s) between two measure spaces which transfer distributions on one space to the other in an optimal way with respect to a given cost function. The way it proves useful in the tracking problem of ensembles is twofold: If the ensemble dynamics is optimal in the sense that it can be cast into the solution of a finite-horizon optimal control problem with a quadratic cost function in terms of the control input, it is shown that the problem may also be translated into a problem of optimal transport between the sampled observations at $t=0,1,2,...,T$ when regarded as probability distributions. Then, Kantorovich formulation of the optimal transport problem will result in computationally efficient solution(s) by transforming the original problem into a linear program with some easily-handled constraints. Furthermore, as we shall see later, for a large class of the cost functions of the corresponding optimal control problem, optimal transport theory will provide a one-shot algorithm to obtain the distributions of the states, in the sense that, once an optimal transport map is identified, the time-varying distribution or the so called \textit{displacement interpolant} is obtained simultaneously.\\

\section{Names, Notations and Conventions}
Throughout this paper, we will use some notations and stick with a few conventions as follows. $\mathbb{N}=\{1,2,3,...\}$ is the set of positive integers. $\mathbb{R}^n$ and $\mathbb{R}^{n \times m}$, for $m,n\in \mathbb{N}$, denote the set of real-valued $n$-vectors and $n$-by-$m$ matrices, respectively. The symbol $||\cdot||$ represents the Euclidean norm. For a matrix $A \in \mathbb{R}^{n \times m}$, $\operatorname{Tr}(A)$ is the trace of $A$, and $\ker A$ is the kernel (or null space) of $A$.

The Borel $\sigma$-algebra on a topological space $\mathcal{X}$ is denoted by $\mathcal{B}(\mathcal{X})$, and the notation $\mathcal{P}(\mathcal{X})$ represents the space of Borel probability measures on $\mathcal{X}$, i.e., probability measures on  $(\mathcal{X},\mathcal{B}(\mathcal{X}))$. The Borel $\sigma$-algebra $\mathcal{B}(\mathbb{R}^n)$, $n\in \mathbb{N}$, is defined with respect to the usual topology on $\mathbb{R}^n$, so $\mathcal{P}(\mathbb{R}^n)$ is well understood. Also, a measure $\mu \in \mathcal{P}(\mathbb{R}^n)$ is called discrete if there exists a countable (finite or infinite) set $A$ with $\mu(\mathbb{R}^n)=\mu(A)$, and the set of all discrete measures in $\mathcal{P}(\mathbb{R}^n)$ is denoted by $\mathcal{P}_d(\mathbb{R}^n)$. For two topological spaces $\mathcal{X}_1$ and $\mathcal{X}_2$, a function $T: \mathcal{X}_1 \rightarrow \mathcal{X}_2$ is called Borel, if it is $\mathcal{B}(\mathcal{X}_1)/ \mathcal{B}(\mathcal{X}_2)$ measurable.

Given a measure space $(X,\mathcal{M}_X,\mu)$, a measurable space $(Y,\mathcal{M}_Y)$, and a measurable function $f: X \rightarrow Y$, the push-forward (or image measure) of $\mu$ under $f$, denoted by $f_{\#}\mu$, is defined by $f_{\#}\mu(S):=\mu(f^{-1}(S))$, for every $S \in \mathcal{M}_Y$, where $f^{-1}(A)=\{x \in X: f(x) \in A\}$, $A \in Y$, is the inverse image of $A$ under $f$. (Note that $f$ need not be one-to-one in order for $f^{-1}(\cdot)$ to be well-defined.) On the other hand, when two measure spaces $(X,\mathcal{M}_X,\mu)$ and $(Y,\mathcal{M}_Y,\nu)$ are considered, a measurable function $f: X \rightarrow Y$ is said to be \textit{measure-~preserving}, if $f_{\#}\mu=\nu$ holds.

For a random variable $\mathbf{x}: \Omega \rightarrow \mathbb{R}^k$, $k \in \mathbb{N}$, on a probability space $(\Omega,\mathcal{F},\mathbb{P})$, the probability law of $\mathbf{x}$, denoted by $\operatorname{law}(\mathbf{x})$, is defined by $\operatorname{law}(\mathbf{x})(B)=\mathbb{P}(\mathbf{x}^{-1}(B))$, for each $B \in \mathcal{B}(\mathbb{R}^k)$, namely, $\operatorname{law}(\mathbf{x}):=\mathbf{x}_{\#}\mathbb{P}$. The notation $\mathbf{y} \sim \mu$, for random variable $\mathbf{y}$ and measure $\mu$, is common throughout the paper, and serves to mean $\operatorname{law}(\mathbf{y})=\mu$. As far as the typeface is concerned, we will write in boldface only the random variables $\mathbf{x}$, and random processes $\mathbf{x}(t)$.

Given two topological spaces $\mathcal{X}$ and $\mathcal{Y}$,  and a measure $\pi \in \mathcal{P}(\mathcal{X} \times \mathcal{Y})$, the \textit{marginals} on $\mathcal{X}$ and $\mathcal{X}$ are defined by $X_{\#}\pi$ and $Y_{\#} \pi$, respectively, where $X: (x,y) \mapsto x$ and $Y: (x,y) \mapsto y$. Also, for given $\mu \in \mathcal{P}(\mathcal{X})$ and $\nu \in \mathcal{P}(\mathcal{Y})$, the symbol $\Pi(\mu, \nu)$ denotes the set of all joint probability measures $\pi \in \mathcal{P}(\mathcal{X} \times \mathcal{Y})$ whose marginals are $\mu$ and $\mu$.
\section{Organization}
The current paper is organized in the following way: chapter 2 first includes a small section on optimal transport theory in order to provide the related notations and concepts for the reference of future developments. Then the setup and the solution of the linear quadratic problem with stochastic end-points which we will be dealt with frequently in the study of ensemble tracking problem will be covered with the viewpoint of optimal transport theory.  Subsequently, chapter 3 addresses controlling, and more extensively, state tracking of ensembles by introducing the notions of ensemble observability and ensemble controllability as well as reviewing the literature therein. While the definitions are adopted for general (time-varying and nonlinear) cases, the focus of this chapter will mainly be on linear time-varying ensembles. Also, in the same chapter, a controllability measure for nonlinear ensembles influenced by control inputs over just a bounded portion of the state space is proposed. And finally, chapter 4 is meant to summarize the conclusions and provide recommendations on the possible extensions and improvements on the solutions of the problems posed here, which we may be adopted in the future.

\renewcommand{\thechapter}{2}

\chapter{Stochastic Control and Optimal Transport}
Consider a linear time-varying control system
\begin{equation}
    \dot{x}(t)=A(t)x(t)+B(t)u(t),  \hspace{10pt} t_0\leq t\leq t_1 \label{s1}
\end{equation}
where $A: [t_0,t_1] \rightarrow \mathbb{R}^{n \times n}$ and $B: [t_0,t_1] \rightarrow \mathbb{R}^{n \times m}$ are continuous, $x(t) \in \mathbb{R}^n$ is the system's state at time $t$, and the piecewise continuous function $u: [t_0,t_0] \rightarrow \mathbb{R}^{m}$ is the control input to the system. In general, a function $f: [a,b] \rightarrow \mathbb{R}^k$ is called \textit{piecewise continuous} if it is continuous except for (at most) finitely many points, and it has finite left and right limits (when they apply) everywhere on $[a,b]$. Further assuming the system \eqref{s1} is controllable over $[t_0,t_1]$, a classic result \cite[Corollary~1]{Kalman} states that, given $x_0, x_1 \in \mathbb{R}^n$, the minimum 
\begin{equation} \label{min_energy}
    \inf_{u \in \mathcal{U}[(x_0,t_0);(x_1,t_1)]} \int_{t_0}^{t_1} ||u(t)||^2\mathrm{d}t
\end{equation}
of the control energy over the set
\begin{equation}
   \mathcal{U}[(x_0,t_0);(x_1,t_1)]=\big\{u:[t_0,t_1] \rightarrow \mathbb{R}^m \big| u \text{ piecewise continuous,} \hspace{5pt} x^u(t_0)=x_0,\hspace{5pt} x^u(t_1)=x_1\big\}
\end{equation}
of admissible control inputs is equal to
\begin{equation} \label{energy_norm}
    ||x_0-\Phi_{(t_0,t_1)}x_1||_{W_{(t_0,t_1)}^{-1}}=(x_0-\Phi_{(t_0,t_1)}x_1)^{\top}W_{(t_0,t_1)}^{-1}(x_0-\Phi_{(t_0,t_1)}x_1),
\end{equation}
where $\Phi_{(t,t')}$ is the state transition matrix \underline{from $t'$ to $t$} associated with the matrix $A(\cdot)$, and 
\begin{equation} \label{ctrb_def}
W_{(t,t')}\overset{\text{def}}{=}\int_{t}^{t'} \Phi_{(t,\tau)}B(\tau)B(\tau)^{\top}\Phi_{(t,\tau)}^{\top}\mathrm{d}\tau
\end{equation}
is the controllability Gramian \underline{from $t$ to $t'$}. Furthermore, this minimum control energy is achieved by the optimal control input (\cite[Theorem~5-bis]{Kalman})
\begin{equation}
    u^{\star}(t)\overset{\text{def}}{=}-B(t)^{\top}\Phi_{(t_1,t)}^{\top} W_{(t_1,t_0)}^{-1}\big( x_1-\Phi_{(t_1,t_0)}x_0\big)=B(t)^{\top}\Phi_{(t_0,t)}^{\top} W_{(t_0,t_1)}^{-1}\big( \Phi_{(t_0,t_1)}x_1-x_0\big),
\end{equation} 
equivalently represented by the feedback law
\begin{equation} \label{u_star}
    u^{\star}(t,x)\overset{\text{def}}{=}-B(t)^{\top}\Phi_{(t_1,t)}^{\top} W_{(t_1,t)}^{-1}\big( x_1-\Phi_{(t_1,t)}x\big)=B(t)^{\top}W_{(t,t_1)}^{-1}\big( \Phi_{(t,t_1)}x_1-x\big).
\end{equation} 
This problem is commonly called the (minimum-energy) \textit{linear quadratic regulator} (LQR) \cite{Liberzon}, the deterministic linear optimal regulator \cite{Kwakernaak}, or the linear-quadratic \cite{Bertsekas, BrysonHo} problem with fixed terminal states.

\begin{remark}
Having shown that an optimal solution of \eqref{min_energy} is represented by the feedback law in \eqref{u_star} and acknowledging that any continuous feedback law $u: [t_0,t_1] \times \mathbb{R}^n \rightarrow \mathbb{R}^m$ gives rise to a continuous control input $\tilde{u}: t \mapsto u(t,x^u(t))$, with $x^u(\cdot)$ the continuously differentiable solution of $\dot{x}(t)=A(t)x(t)+B(t)u(t,x(t))$, we could have replaced the set $ \mathcal{U}[(x_0,t_0);(x_1,t_1)]$ of piecewise continuous control inputs in \eqref{min_energy} with the set 
\begin{equation}
    \overline{\mathcal{U}}[(x_0,t_0);(x_1,t_1)]\overset{\text{def}}{=}\big\{u: [t_0,t_1] \times \mathbb{R}^n \rightarrow \mathbb{R}^m \big| u \text{ continuous,} \hspace{5pt} x^u(t_0)=x_0,\hspace{5pt} x^u(t_1)=x_1\big\}
\end{equation}
of continuous feedback laws, without loss of generality. Then the optimization equivalent to \eqref{min_energy} would be
\begin{equation} \label{min_energy_equiv}
     \inf_{u \in  \overline{\mathcal{U}}[(x_0,t_0);(x_1,t_1)]} \int_{t_0}^{t_1} ||u(t,x^u(t))||^2\mathrm{d}t.
\end{equation}
\end{remark}
Now, if we generalize this setting to the case where the initial and final states are Gaussian random variables $\mathbf{x}_0 \sim \mathcal{N}(m_0, \Sigma_0)$ and $\mathbf{x}_1 \sim \mathcal{N}(m_1, \Sigma_1)$, respectively, then the minimum  of the average  control energy to transfer the system from $(\mathbf{x}_0, t_0)$ to $(\mathbf{x}_1,t_1)$ will~be
\begin{align}
    \inf_{u \in   \overline{\mathcal{U}}[(\mathbf{x}_0,t_0);(\mathbf{x}_1,t_1)]}&\mathbb{E}\left[ \int_{t_0}^{t_1} ||u(t,\mathbf{x}^u(t))||^2\mathrm{d}t\right] \label{LQR_gaussian}\\
    =&\mathbb{E}\left[\inf_{u \in   \overline{\mathcal{U}}[(\mathbf{x}_0,t_0);(\mathbf{x}_1,t_1)]} \int_{t_0}^{t_1} ||u(t,\mathbf{x}^u(t))||^2\mathrm{d}t\right]\\
    =&\mathbb{E}\left[(\mathbf{x}_0-\Phi_{(t_0,t_1)}\mathbf{x}_1)^{\top}W_{(t_0,t_1)}^{-1}(\mathbf{x}_0-\Phi_{(t_0,t_1)}\mathbf{x}_1)\right] \label{eq:22}\\
    =&\mathbb{E}\left[\mathbf{x}_0^{\top}W_{(t_0,t_1)}^{-1}\mathbf{x}_0\right]-2\mathbb{E}\left[(\Phi_{(t_0,t_1)}\mathbf{x}_1)^{\top}W_{(t_0,t_1)}^{-1}\mathbf{\mathbf{x}}_0\right]\notag \\
    &\hspace{30pt}+\mathbb{E}\left[\mathbf{x}_1^{\top}(\Phi_{(t_0,t_1)}^{\top}W_{(t_0,t_1)}^{-1}\Phi_{(t_0,t_1)})\mathbf{x}_1\right] \\
    =& \operatorname{Tr}\big(W_{(t_0,t_1)}^{-1}\mathbb{E}[\mathbf{x}_0 \mathbf{x}_0^{\top}]-2\Phi_{(t_0,t_1)}^{\top}W_{(t_0,t_1)}^{-1}\mathbb{E}\left[\mathbf{x}_0\mathbf{x}_1^{\top}\right] \notag \\
    &\hspace{30pt} +\Phi_{(t_0,t_1)}^{\top}W_{(t_0,t_1)}^{-1}\Phi_{(t_0,t_1)}\mathbb{E}[\mathbf{x}_1 \mathbf{x}_1^{\top}] \big) \label{eq:1111}\\
    =&\operatorname{Tr}\big(W_{(t_0,t_1)}^{-1} (\Sigma_0+m_0m_0^{\top})-2\Phi_{(t_0,t_1)}^{\top}W_{(t_0,t_1)}^{-1}(\Sigma_{0,1}+m_0m_1^{\top}) \notag \\
    &\hspace{30pt} +\Phi_{(t_0,t_1)}^{\top}W_{(t_0,t_1)}^{-1}\Phi_{(t_0,t_1)}(\Sigma_1+m_1m_1^{\top})\big),
\end{align}
where $\Sigma_{0,1}=\mathbb{E}[(\mathbf{x}_0 -m_0)(\mathbf{x}_1-m_1)^{\top}]$, equality \eqref{eq:22} is due to \eqref{energy_norm}, and equality $\eqref{eq:1111}$ is by the identities $\operatorname{Tr}(ABC)=\operatorname{Tr}(CAB)=\operatorname{Tr}(BCA)$. See \cite{Kalman} for a similar result in the spacial case that the system initiated at a zero-mean Gaussian distribution is transferred to the origin.

Motivated by  \cite{Linear Dynamical System}, the next natural step is to study the following stochastic control problem, which is the main problem discussed in this chapter, and will be referred to in later chapters
\begin{subequations} \label{p1}
\begin{align}
    \inf_{u \in\overline{\mathcal{U}}[(\mathbf{x}_0,t_0);(\mathbf{x}_1,t_1)]} &\mathbb{E}\left\{ \int_{0}^1 \frac{1}{2}||u(t,\mathbf{x}^{u}(t))||^2\mathrm{d}t\right\} \\
    \text{subject to:} \hspace{10pt}&\dot{\mathbf{x}}^u(t)=A(t)\mathbf{x}^u (t) + B(t)u(t,\mathbf{x}^{u}(t)), \label{eq:2}\\
    &\mathbf{x}^{u}(0)=\mathbf{x}_0 \sim \mu_0, \hspace{10pt} \mathbf{x}^{u}(1)=\mathbf{x}_1\sim \mu_1,\label{eq:3}
\end{align}
\end{subequations}
where $\mu_0,\mu_1\in \mathcal{P}(\mathbb{R}^n)$. This problem is an extension of the problem \eqref{LQR_gaussian} to the case where the system $\eqref{eq:2}$ is to be steered from an initial state distributed according to the absolutely continuous measure $\mathrm{d}\mu_0(x)=\rho_0(x)\mathrm{d}x$, to the final state distribution according to the absolutely continuous measure $\mathrm{d}\mu_1(x)=\rho_1(x)\mathrm{d}x$. An important assumption here is the controllablity of $(A(\cdot),B(\cdot))$, namely invertibility of the associated controllabilty Gramian $W_{(0,1)}$, as defined in \eqref{ctrb_def}. Then, since we are going to make a connection between the optimal solution of this standard problem and the solution of an optimal transport problem, let us have some preliminaries on optimal transport theory. The main result of this chapter is stated in \thref{Th:1}.

\section{Background on Optimal Transport Theory}
The origin of optimal transport (OT) theory dates back to early 1780's, when the French geometer Gaspard Monge in the seminal memoir \cite{Monge} addressed the problem of finding the optimal way of moving some amount of soil extracted from the ground at a given location to a target location, with respect to a specified transportation cost \cite{Old & New,Evans-Gangbo}. Monge's formulation of the optimal transport problem, given $\mu_0,\mu_1 \in \mathcal{P}(\mathbb{R}^n)$ and measurable function $c: \mathbb{R}^n \times \mathbb{R}^n \rightarrow [0,+\infty]$, is the optimization 
\begin{align} \label{Monge}
  \inf_{T_{\#}\mu_0=\mu_1} \int_{\mathbb{R}^n} c\big(x,T(x)\big) \mathrm{d}\mu_0(x)
\end{align}
over the set of measure-preserving Borel maps $T: \mathbb{R}^n \rightarrow \mathbb{R}^n$ \cite{Gangbo-McCann}. Such maps are usually called \textit{transport maps}, and the minimizers of \eqref{Monge}, if they exist, are called \textit{optimal transport maps} between $\mu_0$ and $\mu_1$. Regardless of the existence of an optimal transport map between $\mu_0$ and $\mu_1$, the optimal value of \eqref{Monge} is denoted by $\mathcal{C}_M(\mu_0, \mu_1)$, and will be called the \textit{Monge optimal cost} between $\mu_0$ and $\mu_1$. Note that, although Monge originally formulated \eqref{Monge} for the special case $c(x,y)=||x-y||$, and for absolutely continuous $\mu_0,\mu_1 \in \mathcal{P}(\mathbb{R}^3)$ \cite[pp.~43]{Old & New}, the literature sometimes associates him with \eqref{Monge} in the general case (see e.g., \cite{Topics}). 

In the 1940's, a relaxed version of the Monge's optimal transport problem was studied \cite{Kantoro1,Kantoro2} by Leonid Kantorovich, Russian mathematician and Nobel Prize laureate in economics. He formulated the problem as the optimization
\begin{align} \label{Kantoro}
  \inf_{\gamma \in \Pi(\mu_0,\mu_1)}\int_{\mathbb{R}^n \times \mathbb{R}^n} c(x,y) \mathrm{d}\gamma(x,y)
\end{align}
over the set $\Pi(\mu_0,\mu_1)$ of Borel probability measures on $\mathbb{R}^n \times \mathbb{R}^n$ with marginals $\mu_0$ and $\mu_1$. Such measures are called \textit{transport plans} between $\mu_0$ and $\mu_1$, and the minimizers of \eqref{Kantoro} are called the \textit{optimal transport plans} between $\mu_0$ and $\mu_1$. Consider that $\Pi(\mu_0,\mu_1)$ is non-empty, since the product measure $\mu_0 \otimes \mu_1 \in \Pi(\mu_0,\mu_1)$, and that regardless of the existence of an optimal transport plan between $\mu_0$ and $\mu_1$, the optimal value of \eqref{Kantoro} is denoted by $\mathcal{C}_K(\mu_0, \mu_1)$ and called the \textit{Kantorovich optimal cost} between $\mu_0$ and $\mu_1$. Also note that in \eqref{Kantoro} the objective function is linear in $\gamma$, and the optimization domain $\Pi(\mu_0,\mu_1)$ is convex. So, it can be shown that under appropriate assumptions on the function $c$, existence of an optimal transport plan is guaranteed (see e.g., \cite[Theorem 1.3]{Topics} and \cite[Theorem 2.1]{Ambrosio}). Also, for a proof of the existence of an optimal transport plan associated with the cost $c(x,y)=||x-y||^2$, see \cite[Proposition 2.1]{Topics}.

Kantorovich's optimization \eqref{Kantoro} for the cost $c(x,y)=||x-y||^p$, $p \in [1, \infty)$, induces a metric, called \textit{p\nobreakdash-Wasserstein metric} (distance), defined by
\begin{equation}
    W_p(\mu_0,\mu_1)\overset{\text{def}}{=}\left(\inf_{\gamma \in \Pi(\mu_0,\mu_1)}\int_{\mathbb{R}^n \times \mathbb{R}^n} \left\Vert y-x\right\Vert^p \mathrm{d}\gamma(x,y)\right)^{1/p},
\end{equation}
on the space $\mathcal{P}_p(\mathbb{R}^n)=\{\mu \in \mathcal{P}(\mathbb{R}^n): \int_{\mathbb{R}^n} ||x-x_0||^p\mathrm{d}\mu(x)<\infty, \hspace{5pt}x_0 \in \mathbb{R}^n\}$ of Borel probability measures with finite moment of order $p$ \cite[Theorem 7.3]{Topics}.

\begin{remark} \thlabel{relaxation}
The claim that \eqref{Kantoro} is a relaxed version of \eqref{Monge} can be established in two steps. First, consider that, if $T_{\#}\mu_0=\mu_1$ holds for a map $T$, then the induced measure $\gamma_{_T}=\tilde{T}_{\#}\mu_0$, with $\tilde{T}(x)=(x,T(x))$,
does indeed belong to $\Pi(\mu_0,\mu_1)$, since for every $A_1, A_2\in \mathcal{B}(\mathbb{R}^n)$ we have
\begin{subequations}
\begin{align}
    \gamma_{_T}(A_1 \times \mathbb{R}^n)=&\mu_0 \left(\{x\in \mathbb{R}^n: \left(x,T(x)\right)\in A_1 \times \mathbb{R}^n\}\right)=\mu_0(A_1)\\
   \gamma_{_T}(\mathbb{R}^n \times A_2)=&\mu_0 \left(\{x\in \mathbb{R}^n: \left(x,T(x)\right)\in \mathbb{R}^n \times A_2 \}\right)=\mu_0(T^{-1}(A_2))=T_{\#}\mu_0(A_2)=\mu_1(A_2).
\end{align}
\end{subequations}
Second, the change of variables theorem \cite[Theorem 3.6.1]{Bogachev} implies 
\begin{equation}
    \int_{\mathbb{R}^n \times \mathbb{R}^n} c(x,y)\mathrm{d}\gamma_{_T}(x,y)=\int_{\mathbb{R}^n} c(\tilde{T}(x))\mathrm{d}\mu_0(x)=\int_{\mathbb{R}^n} c(x,T(x))\mathrm{d}\mu_0(x),
\end{equation}
which means that the objective function in \eqref{Monge} evaluated for the transport map $T$ coincides with the value of the objective function in \eqref{Kantoro} for the transport plan $\tilde{\gamma}$. 
\end{remark}
A consequence following from \thref{relaxation} is that, in general, the optimal value of Kantorovich's problem \eqref{Kantoro} provides a lower bound for that of Monge's problem \eqref{Monge}, namely,
\begin{equation} \label{K<M}
    \mathcal{C}_K(\mu_0, \mu_1) \leq \mathcal{C}_M(\mu_0, \mu_1).
\end{equation}
However, a result by Gangbo and McCann \cite{Gangbo-McCann} for a specific class of costs of the form $c(x,y)=h(x-y)$, including the cost $c(x,y)=||x-y||^2$, implies that \eqref{Kantoro} admits a unique optimal transport plan $\gamma_{_T}$ which is induced by a transport map $T$ between $\mu_0$ and $\mu_1$, provided that $\mu_0$ and $\mu_1$ are absolutely continuous with respect to Lebesgue measure, and $\mathcal{C}_K(\mu_0, \mu_1)<\infty$. Since this result, especially applied to $c(x,y)=||x-y||^2$, will be of great use throughout the paper, we will summarize it in the following theorem.

\begin{theorem}[When $\mathcal{C}_K(\mu_0, \mu_1) = \mathcal{C}_M(\mu_0, \mu_1)$] \thlabel{M=K}
Given the cost $c(x,y)=||x-y||^2$ and Borel probability measures $\mu_0,\mu_1 \in \mathcal{P}(\mathbb{R}^n)$ in \eqref{Monge} and \eqref{Kantoro}, if $\mu_0$ is absolutely continuous with respect to Lebesgue measure, and $W_2(\mu_0, \mu_1)<\infty$, then there exists a unique optimal transport plan between $\mu_0$ and $\mu_1$ defined by $\gamma_{_T}=\tilde{T}_{\#}\mu_0$, with $\tilde{T}(x)=(x,T(x))$, induced by a transport map $T$ between $\mu_0$ and $\mu_1$. In that case, 
\begin{equation} \label{Mon=Kan}
    \inf_{T_{\#}\mu_0=\mu_1} \int_{\mathbb{R}^n} ||T(x)-x||^2 \mathrm{d}\mu_0(x)=\inf_{\gamma \in \Pi(\mu_0,\mu_1)}\int_{\mathbb{R}^n \times \mathbb{R}^n} \left\Vert y-x\right\Vert^2 \mathrm{d}\gamma(x,y).
\end{equation}
\end{theorem}

\begin{proof}
The existence and uniqueness of the optimal transport plan is a direct implication of Theorem 3.7 in \cite{Gangbo-McCann} for $c(x,y)=||x-y||^2$. Also, equation \eqref{Mon=Kan} follows from \eqref{K<M} and the fact that the optimal transport plan $\gamma_{_T}$ is indeed induced by a transport map $T$.
\end{proof}
Whenever, in the settings of the following sections, the hypotheses of \thref{M=K} are valid, we may interchangeably use the name 2-Wasserstein distance for the square root of the optimal value of either equivalent formulations in \eqref{Mon=Kan}.

\begin{remark} \thlabel{r:1}
An elegant result by Benamou and Brenier \cite{fluid dynamics} of which we will make great use in our developments later, provides a connection between $2$-Wasserstein distance of two measures, say $\mu_0$ and $\mu_1$, with corresponding densities $\rho_0$ and $\rho_1$, and a minimization problem over the pairs $(\rho,u)$ of time-varying densities and velocity fields, for a fluid with initial and final densities $\rho_0$ and $\rho_1$, respectively: 
\begin{subequations}
\begin{align}
   W_2(\mu_0,\mu_1)^2= &\inf_{(\rho,u)} \int_{\mathbb{R}^n} \int_0^1 ||u(t,x)||^2 \rho(t,x) \mathrm{d}t \mathrm{d}x &\\
    \text{subject to :} \hspace{10pt} &\frac{\partial \rho}{\partial t}+\nabla \cdot (\rho u)=0 \label{cont:1} &\\
    &\rho(0,x)=\rho_0(x), \hspace{10pt} \rho(1,x)=\rho_1(x), & \forall x \in \mathbb{R}^n \label{cont:2}.
\end{align}
\end{subequations}
In  fluid dynamics and continuum mechanics literature, the equation $\eqref{cont:1}$ is commonly called the \textit{"continuity equation"}, which indeed represents the mass conservation law in the differential form (\cite{Goldstein}). So, from now on, we will refer to any pair $(\rho,u)$ satisfying $\eqref{cont:1}$ as one satisfying the continuity condition. Also, note that, the continuity equation is in general a necessary condition for the flow of continua. However, under some mild conditions, which we will assume later, any pair $(\rho,u)$ satisfying the continuity condition does represent the unique evolution of the density of the flow subject to the velocity field $u$. The way this fluid dynamic representation proves useful will be clear in the next chapter, especially section 3.2.
\end{remark}

\section{Solution of the Stochastic Control Problem by Optimal Transport Theory}
Back to the problem $\eqref{p1}$, let us again consider the optimal feedback control law
\begin{equation}
    u^{\star}\big(t,x;y\big)=-B(t)^{\top}\Phi_{(1,t)}^{\top} W_{(1,t)}^{-1}\big( y-\Phi_{(1,t)}x\big), 
\end{equation}
(see \eqref{u_star}) for the linear quadratic problem over the time interval $[0,1]$ with fixed terminal states $x,y \in \mathbb{R}^n$. First, by reminding ourselves that $\mathcal{U}$ is the set of (piecewise) continuous feedback laws, note that for any $\mathbf{y} \sim \mu_1$, such that $\mathbf{y}=\mathbf{x}^{u}(1)$, $u \in \mathcal{U}$, the feedback control law $u^{\star}(\cdot,\cdot;\mathbf{y})$ is admissible, i.e., $(t,x) \longmapsto u^{\star}\big(t,x;\mathbf{y}\big)$ belongs to $\mathcal{U}$. Then, for fixed $\omega \in \Omega$, where $\Omega$ is some sample space on which the processes $\mathbf{x}^u$ are defined, consider that we have 
\begin{equation}
\int_{0}^1 \frac{1}{2}\left\Vert u^{\star}\big(t,\mathbf{x}^{u^{\star}}(t,\omega);\mathbf{x}^{u}(1,\omega)\big)\right\Vert^2\mathrm{d}t \leq \int_{0}^1 \frac{1}{2}\left\Vert u\big(t,\mathbf{x}^{u}(t,\omega)\big)\right\Vert^2\mathrm{d}t,
\end{equation}
by construction. This implies the feedback control law $(t,x) \longmapsto u^{\star}\big(t,x;\mathbf{x}^u(1)\big)$ outperforms $u$. Hence, the problem {\eqref{p1}} reduces to 
\begin{align}
    \inf_{u \in \mathcal{U}^{\star}} \mathbb{E}\left\{ \int_{0}^1 \frac{1}{2}\left\Vert u\big(t,\mathbf{x}^{u}(t)\big)\right\Vert^2\mathrm{d}t\right\},
\end{align}
where $\mathcal{U}^{\star}=\left\{u \in \mathcal{U}: u(t,x)=u^{\star}\big(t,x;T(\mathbf{x}(0))\big),\hspace{4 pt} T \text{ Borel }\right\}$. Then, we will have
\begin{align}
    &\inf_{u \in \mathcal{U}} \mathbb{E}\left\{ \int_{0}^1 \frac{1}{2}\left\Vert u\big(t,\mathbf{x}^{u}(t)\big)\right\Vert^2\mathrm{d}t\right\}\\
    =&\inf_{u \in \mathcal{U}^{\star}} \mathbb{E}\left\{ \int_{0}^1 \frac{1}{2}\left\Vert u\big(t,\mathbf{x}^{u}(t)\big)\right\Vert^2\mathrm{d}t\right\}\\
    =&\inf_{\{T: T(\mathbf{x}(0))\sim \mu_1\}} \mathbb{E}\left\{ \int_{0}^1 \frac{1}{2}\left\Vert u^{\star}\big(t,\mathbf{x}^{u^{\star}}(t);T(\mathbf{x}(0))\big)\right\Vert^2\mathrm{d}t\right\}\\
    =&\inf_{\{T: T(\mathbf{x}(0))\sim \mu_1\}}\mathbb{E}\left\{ \frac{1}{2}\big[\mathbf{x}(0)-\Phi_{(0,1)}T(\mathbf{x}(0))\big]^{\top}W_{(0,1)}^{-1}\big[\mathbf{x}(0)-\Phi_{(0,1)}T(\mathbf{x}(0)\big]\right\},\\
    =&\inf_{T_{\#}\mu_0=\mu_1}\int_{\mathbb{R}^n} \frac{1}{2}\big[\mathbf{x}(0)-\Phi_{(0,1)}T(\mathbf{x}(0))\big]^{\top}W_{(0,1)}^{-1}\big[\mathbf{x}(0)-\Phi_{(0,1)}T(\mathbf{x}(0)\big]\mathrm{d}\mu_0(x),\\
    =&\inf_{T_{\#}\mu_0=\mu_1}\int_{\mathbb{R}^n} c\big(x,T(x)\big)\mathrm{d}\mu_0(x),\label{eq:100}
\end{align}
where $c: \mathbb{R}^n \times \mathbb{R}^n \rightarrow \mathbb{R}$, with $c(x,y)=\frac{1}{2}\big(x-\Phi_{(0,1)}y\big)^{\top}W_{(0,1)}^{-1}\big(x-\Phi_{(0,1)}y\big)$. However, using the transformation $(\hat{x},\hat{y})=(W_{(1,0)}^{-1/2}\Phi_{(1,0)}x,W_{(1,0)}^{-1/2}y) \sim (\hat{\mu}_0,\hat{\mu}_1)$, the infimum in $\eqref{eq:100}$ can be cast into
\begin{equation}
    \inf_{\hat{T}_{\#}\hat{\mu}_0=\hat{\mu}_1}\int_{\mathbb{R}^n} \left\Vert \hat{T}(\hat{x})-\hat{x}\right\Vert^2 \mathrm{d}\hat{\mu}_0(\hat{x})=\inf_{\hat{\pi} \in \Pi(\hat{\mu}_0,\hat{\mu}_1)}\int_{\mathbb{R}^n \times \mathbb{R}^n}  \left\Vert \hat{y}-\hat{x}\right\Vert^2 \mathrm{d}\hat{\pi}(\hat{x},\hat{y}),\label{eq:200}
\end{equation}
which is the square of the 2-Wasserstein distance $W_2(\hat{\mu}_0,\hat{\mu}_1)$ between $\hat{\mu}_0=[W_{(1,0)}^{-1/2}\Phi_{(1,0)}]_{\#}\mu_0$ and $\hat{\mu}_1=[W_{(1,0)}^{-1/2}]_{\#}\mu_1$. Thus, the minimization in $\eqref{eq:100}$ can be rewritten, again in terms of the cost function $c$, as 
\begin{equation}
\inf_{\pi \in \Pi(\mu_0,\mu_1)}\int_{\mathbb{R}^n \times \mathbb{R}^n}  c(x,y) \mathrm{d}\pi(x,y). \label{eq:201}
\end{equation}
While equivalent representations $\eqref{eq:200}$ of the original stochastic control problem, especially the one searching over joint distributions, are closer to the Kantorovich optimal transport formulation, and hence, benefit from computational efficiencies of linear programming, representation $\eqref{eq:201}$ involves the given initial and final distributions explicitly, a fact which brings notational ease for us in some of the developments in the sequel. 
Now, it is timely to summarize what we have developed so far, in a theorem:
\begin{theorem}\thlabel{Th:1}
The optimal value of the minimization problem $\eqref{p1}$ is equal to the square $W_2(\hat{\mu}_0,\hat{\mu}_1)^2$ of the 2-Wasserstein distance between $\hat{\mu}_0=[W_{(1,0)}^{-1/2}\Phi_{(1,0)}]_{\#}\mu_0$ and $\hat{\mu}_1=[W_{(1,0)}^{-1/2}]_{\#}\mu_1$. Furthermore, this value is achieved by the feedback control law
\begin{equation}
    u^{\star}\big(t,\mathbf{x};T(\mathbf{x}(0))\big)=B(t)^{\top}\Phi_{(1,t)}^{\top} W_{(1,t)}^{-1}\big( T(\mathbf{x}(0))-\Phi_{(1,t)}\mathbf{x\big)},
\end{equation}
where $T$ is the solution of the Monge's optimal transport problem between $\hat{\mu}_0$ and $\hat{\mu}_1$.
\end{theorem}

\begin{remark} \thlabel{Gaussian_remark}
As a benchmark, the stochastic control problem \eqref{p1} with Gaussian end-point distributions
$\mu_0 \longleftrightarrow (\nu_0, \Sigma_0)$ and $\mu_1 \longleftrightarrow (\nu_1, \Sigma_1)$ are studied in \cite{Linear Dynamical System}. Also, besides deriving a closed form solution for the optimal control input $u$ (compare with the one in \thref{Th:1} where $u$ is in terms of the optimal map $T$), the time-varying distribution, a.k.a. displacement interpolation, $\mu_t$, $0 \leq t \leq 1$, of the optimal state trajectory $x^u(t)$ is obtained. In fact, by considering that the cost function of the optimal transport problem associated with \eqref{p1} is quadratic (see \eqref{eq:100}), it is shown in \cite{Linear Dynamical System} that $\mu_t$ are Gaussian $\mathcal{N}(\nu_t, \Sigma_t)$ with mean and convariance 
\begin{subequations}
\begin{align}
    \nu_t=&\hat{\Phi}_{(t,0)}\nu_0+\int_0^{t} \hat{\Phi}_{(t,\tau)}B(\tau)B(\tau)^{\top}m(\tau)\mathrm{d}\tau,\\
    \Sigma_t=&\Phi_{(t,0)}W_{(0,t)}\Sigma_0^{-1/2}\Bigg[-\Sigma_0^{1/2}W_{(0,1)}^{-1}\Sigma_0^{1/2}+\left(\Sigma_0^{1/2}W_{(0,1)}^{-1}\Phi_{(0,1)}\Sigma_1 \Phi_{(0,1)}^{\top} W_{(0,1)}^{-1}\Sigma_0^{1/2}\right)^{1/2} \notag\\
    & \hspace{175pt} +\Sigma_0^{1/2}W_{(0,t)}^{-1}\Sigma_0^{1/2}\Bigg]^{1/2} \Sigma_0^{-1/2}W_{(0,t)}\Phi_{(t,0)}^{\top},
\end{align}
\end{subequations}
where $W_{(t_1,t_2)}$ is the controllability Gramian (cf. \eqref{ctrb_def}) from $t_1$ to $t_2$ associated with $\Phi$, and 
\begin{equation}
    m(t)=\hat{\Phi}_{(0,t)}^{\top} \hat{W}_{(1,0)}^{-1}\left(\hat{\Phi}_{(0,1)}\nu_1-\nu_0\right).
\end{equation}
In the equation above, $\hat{W}_{(t_1,t_2)}$ is the controllability Gramian from $t_1$ to $t_2$ associated with $\hat{\Phi}$, and $\hat{\Phi}$ itself is the state transition matrix related to $\left(A(t)-B(t)B(t)^{\top}K(t)\right)$, such that $K(t)$ satisfies the following differential Riccati equation with initial conditions
\begin{subequations}
\begin{align}
    \dot{K}(t)=&-A(t)^{\top}K(t)-K(t)A(t)+K(t)B(t)B(t)^{\top}K(t)\\
    K(0)=&\Sigma_0^{-1/2}\Bigg[\Sigma_0^{1/2} W_{(0,1)}^{-1}\Sigma_0^{1/2}-\left(\Sigma_0^{1/2}W_{(0,1)}^{-1}\Phi_{(0,1)}\Sigma_1\Phi_{(0,1)}^{\top}W_{(0,1)}^{-1}\Sigma_0^{1/2}\right)^{1/2}\Bigg]\Sigma_0^{-1/2}.
\end{align}
\end{subequations}
We will refer to this result later in Chapter 3 where the state tracking problem with Gaussian output distributions is studied.
\end{remark}

Motivated by \thref{r:1}, and with the goal of computational feasibility, we are interested in obtaining the fluid dynamic version of the problem $\eqref{p1}$. An elementary derivation is provided in \cite{Linear Dynamical System}. But the full derivation with the required assumptions is presented in the following theorem:

\begin{theorem} \thlabel{Th:2}
For the stochastic control problem $\eqref{p1}$, we have the equivalent fluid-dynamic representation
\begin{equation*}
    \inf_{u \in \mathcal{U}} \mathbb{E}\left\{ \int_{0}^1 \frac{1}{2}||u(t,\mathbf{x}^{u}(t))||^2\mathrm{d}t\right\}=\inf_{(\rho,u) \in \mathcal{W}} \int_{\mathbb{R}^n} \int_0^1 \frac{1}{2}||u(t,x)||^2 \rho(t,x) \mathrm{d}t \mathrm{d}x,
\end{equation*}
where $\mathcal{W}$ is the set of pairs $(\rho,u)$ satisfying $\eqref{cont:1}$ and $\eqref{cont:2}$, modified as follows:
\begin{subequations}
\begin{align}
    &\frac{\partial \rho}{\partial t}(t,x)+\nabla_{x} \cdot \big[\rho(t,x)\big(A(t)x + B(t)u(t,x)\big)\big]=0, \hspace{10pt} &\forall (t,x),\\
    &\rho(0,x)=\rho_0(x), \hspace{10pt} \rho(1,y)=\rho_1(y), \hspace{10pt} &\forall (x,y).
\end{align}
\end{subequations}
\end{theorem}

\begin{proof}
First, we would like to show for every $u \in \mathcal{U}$, there exists $\rho$ such that $(\rho,u) \in \mathcal{W}$. To that end, let $u \in \mathcal{U}$ and define the vector field $V_u(\mathbf{x},t)=A(t)\mathbf{x}+B(t)u(t,\mathbf{x})$, associated with $u$. Then, taking a look at $V_u$ along the trajectories of $\eqref{eq:2}$ with $u$ as the control law, we have $V_u(\mathbf{x}(t),t)=A(t)\mathbf{x}(t)+B(t)u(t,\mathbf{x}(t))=\dot{\mathbf{x}}(t)$. This implies $V_u$ is indeed the "velocity" vector field generated by $u$ through the dynamics $\eqref{eq:2}$. Moreover, uniqueness of the solution lets us define the dynamical system $\Pi^u$ induced by $\eqref{eq:2}$, such that $\Pi^u(x,t)$ is the state at time $t$ of the system $\eqref{eq:2}$ initiated at $x \in \mathbb{R}^n$. Thus, it is easily verified that distribution $\mu_t$ of $\mathbf{x}^u(t)$, $0\leq t\leq 1$, is uniquely determined by $\mu_t={\Pi_t}_{\#}\mu_0$, or equivalently, $\mu_t(A)=\mu_0({\Pi_t}^{-1}(A))$, where $\Pi_t(x)=\Pi^u(x,t)$. Correspondingly, the density $\rho_t$ (with respect to Lebesgue measure $m$) of $\mathbf{x}^u(t)$ is uniquely defined $m-$a.e. So, taking $(t,x) \longmapsto \rho(t,x)$, with $\rho(t,\cdot)=\rho_t(\cdot)$, we observe $(\rho, V_u)$ necessarily satisfies the continuity equation $\frac{\partial \rho}{\partial t}(t,x)+\nabla_{x} \cdot (\rho V_u)(t,x)=0$, implying $(\rho,u) \in \mathcal{W}$. Next, we aim to show that any $u$ such that $(u,\rho) \in \mathcal{W}$, for some $\rho$, does belong to $\mathcal{U}$. In order to do so, first note that $V_u$, as shown earlier, is the velocity vector field  for trajectories of $\eqref{eq:2}$, and the pair $(\rho, V_u)$ satisfies the continuity equation, by hypothesis. Also, $V_u$ naturally induces a unique dynamical system $\Pi^u$, based on the solution trajectories of $\eqref{eq:2}$, which itself leads to the unique distributions $\mu_t$ of $\mathbf{x}^u(t)$, $0\leq t\leq 1$. Hence, an argument similar to that in the first part of the proof implies that there exists a unique density $\Tilde{\rho}$, $m-$a.e., such that $(\Tilde{\rho}, V_u)$ satisfies the continuity equation together with $\eqref{cont:2}$. This readily implies $\Tilde{\rho}=\rho$, $m-$a.e., and $u \in \mathcal{U}$. Finally, the fact that $\mathbb{E}\left\{ \int_{0}^1 ||u(t,\mathbf{x}^{u}(t))||^2\mathrm{d}t\right\}= \int_{\mathbb{R}^n} \int_0^1 ||u(t,x)||^2 \rho(t,x) \mathrm{d}t \mathrm{d}x$, with $\rho(t,\cdot)$ as the density of $\mathbf{x}^u(t)$, concludes the proof.
\end{proof}

\renewcommand{\thechapter}{3}

\chapter{Controlling and State Tracking of Ensembles}
In this chapter, we would like to study \textit{ensembles} of systems, agents, particles, or the like, which are, from a systems theoretic point of view, indistinguishable individually, and governed by the same dynamics. More specifically, in this setting, there exist some quantities of interest representing the states of the individual systems which evolve under the same dynamics for each system, possibly with different initial conditions, but are observable only in some aggregate form. The control community has coined the term ensemble control for the related problems on this topic \cite{Br-Kh}-\cite{Li}. Suppose, for instance, that we deal with an ensemble of particles, indexed by the set $I$, all influenced by the controlled dynamics
\begin{subequations} \label{gen_dyn}
\begin{align} 
    \dot{\mathbf{x}}^{(i)}(t)&=f(t,\mathbf{x}^{(i)}(t),u(t)),\hspace{10pt}, \mathbf{x}^{(i)}(0)=\mathbf{x}^{(i)}_0 \\
    \mathbf{y}^{(i)}(t)&=g(t,\mathbf{x}^{(i)}(t)),
\end{align}
\end{subequations}

where $\mathbf{x}^{(i)}(t) \in \mathbb{R}^n$ and $\mathbf{y}^{(i)}(t) \in \mathbb{R}^m$ are the state and the output of the $i-$th particle, $i \in I$, respectively, and $u(t) \in \mathbb{R}^p$ is the control input, at time $t\geq 0$. First of all, note that the same control input is applied to all of the systems, which indeed causes a great challenge when it comes to the steering of the ensemble to desirable states. Due to this limitation, ensemble control is sometimes referred to as broadcast (feedback) control, see e.g. \cite{Broadcast1}, \cite{Broadcast2}.  Furthermore, in this setting, neither the states nor the outputs of the particles are accessible individually. Instead, the collection $Y(t)=\{\mathbf{y}^{(i)}(t)\}_{i \in I}$ of output measurements of the whole ensemble is available at some time instances $t=0,1,2,...,T$. So, while our capabilities in controlling each particle by a separate input signal is limited, our observation of the states of the individual particles is just through the lens of aggregate behaviour of the whole ensemble. For the reasons that should be clear now, and in accordance with \cite{Ensemble Observability}, we shall call the common dynamics $\eqref{gen_dyn}$ within the ensemble under study the \textit{structural system} or the \textit{structural dynamics} of the ensemble.

Further, let us assume that the way output measurements, for each time $t \in [0,T]$, of the systems are distributed over $\mathbb{R}^m$ is characterized by some probability measure $\mu_t \in \mathcal{P}(\mathbb{R}^m)$, where $\mathcal{P}(X)$ is the set of probability measures on $X$, such that for any Borel set $A \subset \mathbb{R}^m$, $\mu_t(A)$ represents the portion of the ensemble with output measurements lying in $A$ at time $t$. Correspondingly, $\hat{\mu}_t \in \mathcal{P}(\mathbb{R}^n)$ can be understood as the way states of the systems at time $t$ are distributed  over $\mathbb{R}^n$. Now, using this characterization, we aim to interpolate the distributions $\hat{\mu}_t$, $t \in [0,T]$, based solely on the available output distributions $\{\mu_t\}_{t=0,1,...,T}$. This problem is referred to as the \textit{state tracking probelem for ensembles}, and has recently been studied in \cite{Sampled Observability} for the so-called linear discrete  ensembles.

\section{Observability of Ensembles}

The state tracking problem for ensembles is naturally entangled with the notion \textit{ensemble observability} defined in \cite{Ensemble Observability}, as the ability of uniquely recovering $\hat{\mu}_0$ given $\{\mu_t\}_{t\geq 0}$ and $\{u(t)\}_{t\geq 0}$, as compared with the classic notion of observability for a linear dynamical system \cite{Kalman1960, Kalman1963}, which is the ability of uniquely recovering the deterministic initial state $x(0)$ from the knowledge of $\{y(t)\}_{t\geq 0}$ and $\{u(t)\}_{t\geq 0}$.

\begin{definition}[Ensemble Observability] \thlabel{def_ens_obs}
The ensemble \eqref{gen_dyn} is said to be observable if
\begin{equation} \label{ens_obs_def}
    \hat{\mu}_0^{(1)} \neq \hat{\mu}_0^{(2)} \Longrightarrow \exists t \geq 0: \hspace{10pt} \mu_t^{(1)} \neq \mu_t^{(2)}
\end{equation}
for every $\hat{\mu}_0^{(1)},\hat{\mu}_0^{(2)} \in \mathcal{P}(\mathbb{R}^{n})$, the space of all probability measures on $\mathbb{R}^n$.
\end{definition}

In fact, by posing certain inverse problems in biology, \cite{Ensemble Observability} turns out to have provided the first exposition of the connections between the systems theoretical concept of observability and mathematical tomography. In the same study, motivated by the connections with mathematical tomography, an algebro-geometric characterization of ensemble observability for linear systems has been derived, which we will recover in the next section by another approach.

Let us now consider an ensemble with linear time-varying structural dynamics
\begin{subequations} \label{lin_sys1}
\begin{align} 
    \dot{\mathbf{x}}^{(i)}(t)&=A(t)\mathbf{x}^{(i)}(t)+B(t)u(t) \notag \\
    \mathbf{y}^{(i)}(t)&=C(t)\mathbf{x}^{(i)}(t).
\end{align}
\end{subequations}

Now, motivated by a similar argument in the classical observability of linear systems, we are curious to know if the the piece of information $\{u(t)\}_{t\geq 0}$ is ever material to the notion of ensemble observability. Note that $\mathbf{y}^{(i)}(t)=C(t)\Phi(t,0)\mathbf{x}^{(i)}(0)+C(t)\int_0^{t}\Phi(t,\tau)B(\tau)u(\tau)\mathrm{d}\tau$, with $\Phi(t_2,t_1)$ the state transition matrix from $t_1$ to $t_2$ associated with $A(\cdot)$, and let $\bar{Y}(t):=\{\bar{\mathbf{y}}^{(i)}(t)\}_{i \in I}$, where $\bar{\mathbf{y}}^{(i)}(t)=\mathbf{y}^{(i)}(t)-C(t)\int_0^{t}\Phi(t,\tau)B(\tau)u(\tau)\mathrm{d}\tau$ denotes the uncontrolled (zero-input) response of the system $i$ at time $t$. Further, assume that $\bar{Y}(t)$ is characterized by  some probability measure $\bar{\mu}_t \in \mathcal{P}(\mathbb{R}^m)$, in a similar fashion to $Y(t)$, denoted by $\bar{Y}(t) \longleftrightarrow \bar{\mu}_t$. Obviously, if the controlled ensemble \eqref{lin_sys1} is observable in the sense of $\eqref{ens_obs_def}$, then the uncontrolled ensemble is also observable in the same sense, i.e.
\begin{equation} \label{weak_def}
    \hat{\mu}_0^{(1)} \neq \hat{\mu}_0^{(2)} \Longrightarrow \exists t \geq 0: \hspace{10pt} \bar{\mu}_t^{(1)} \neq \bar{\mu}_t^{(2)},
\end{equation}
since the uncontrolled case is a special case of the controlled case with $u\equiv 0$. On the other hand, we want to show if simply the uncontrolled ensemble is observable, (namely, the seemingly weaker condition \eqref{weak_def} holds) then so is the controlled ensemble. To this end, take two different initial state distributions $ \hat{\mu}_0^{(1)} \neq \hat{\mu}_0^{(2)}$ for the ensemble. Then, \eqref{weak_def} implies $\bar{Y}^{(1)}(t) \neq \bar{Y}^{(2)}(t)$, for some $t \geq 0$, as $ \bar{Y}^{(n)}(t) \longleftrightarrow \bar{\mu}^{(n)}_t$, $n=1,2$. This readily concludes $\mu_t^{(1)} \neq \mu_t^{(2)}$, since that $Y^{(n)}(t) \longleftrightarrow \mu_t^{(n)}$, and that $Y^{(n)}(t)=\bar{Y}^{(n)}(t)+C(t)\int_0^{t}\Phi(t,\tau)B(\tau)u(\tau)\mathrm{d}\tau$, $n=1,2$. 

What the argument above reveals is that observability of a linear time-varying ensemble does not have to do with the information $\{u(t)\}_{t\geq 0}$. More clearly, of course in order to obtain the possible initial state distribution(s) $\hat{\mu}_0$ of the ensemble \eqref{lin_sys1} the knowledge of $\{u(t)\}_{t\geq 0}$ is required along with $\{\mu_t\}_{t \geq 0}$, uniqueness of $\hat{\mu}_0$ is dependent solely on the uncontrolled structural dynamics characterized by $(A(\cdot),C(\cdot))$. Hence, 

Given the above remark, it will also be shown in the following section that there is a close connection (indeed, equivalence) between ensemble observability and the classical observability of the structural system in the special case of linear discrete ensembles, as defined in the sequel.

\subsection{Observability of LTI Discrete Ensembles}
Suppose the ensemble under study consists of a countable (finite or infinite) set of continuous-time LTI systems indexed by $i \in \mathbb{N}$, with identical state and output dynamics, and (possibly) different initial states
\begin{align} \label{sys1}
    \dot{\mathbf{x}}^{(i)}(t)=&A\mathbf{x}^{(i)}(t), \hspace{10pt} \mathbf{x}^{(i)}(0)=\mathbf{x}^{(i)}_0\\
    \mathbf{y}^{(i)}(t)=&C\mathbf{x}^{(i)}(t), \notag
\end{align}

where $A \in \mathbb{R}^{n \times n}$ and $C \in \mathbb{R}^{p \times n}$. Now, if the initial state distribution of the systems within ensemble is characterized by some discrete probability measure $\mu_0$, then clearly the output $\mathbf{y}(t)$ will be distributed according to the push-forward $L_{\#}\mu_0$ of the measure $\mu_0$ under linear mapping $L: \mathbb{R}^n \rightarrow \mathbb{R}^p$, with $L:x \mapsto Ce^{At}x$. In other words, $L_{\#}\mu_0(\mathcal{V}):=\mu_0(L^{-1}(\mathcal{V}))=\mu_0\left((Ce^{At})^{-1}(\mathcal{V})\right)$, for every $\mathcal{V} \in \mathcal{B}(\mathbb{R}^p)$, with $\mathcal{B}(\mathbb{R}^p)$ the Borel sigma-algebra on $\mathbb{R}^p$, and $L^{-1}(\mathcal{V})$ the inverse image of $\mathcal{V}$ under mapping $L$. When there is no fear of ambiguity, and for better reference to the time, we may use the notation $(Ce^{At})_{\#} \mu_0$ instead of $L_{\#} \mu_0$, for output distribution at time $t$. Thus, the property \thref{def_ens_obs} requires is that given all output distributions $\{(Ce^{At})_{\#} \mu_0: t\geq 0\}$ we would be able to construct the initial distribution $\mu_0$ uniquely. Returning to the goal of this section, we are supposed to answer the question on the relationship between observability of the ensemble \eqref{sys1} and the observability of the pair $(A,C)$ associated with the ensemble structural dynamics. And in fact, the ensuing analysis will show that for the more general case of LTI ensembles with arbitrary state and output distributions, i.e. ensembles not necessarily consisting of countable systems, the former condition is stronger than the latter.

To make our claim mathematically rigorous, let us first be reminded that the unobservable subspace $\mathcal{UO}(A,C)$ of an LTI system $(A,C)$ is defined as 
\begin{equation}
   \mathcal{UO}(A,C):= \bigcap_{t \geq 0} \ker Ce^{At},
\end{equation}
and observability of $(A,C)$ is equivalent to $\mathcal{UO}(A,C)=\{0\}$ (\cite{Hespanha}). Now, let us assume that $(A,C)$ is not observable which implies, by the above characterization, that there exist $\nu \neq 0$ such that $Ce^{At}\nu=0$, for all $t \geq 0$. Then given two initial state distributions $\mu_0^{(1)}, \mu_0^{(2)} \in \mathcal{P}_d(\mathbb{R}^{n})$, with $\mu_0^{(2)}(\mathcal{W})=\mu_0^{(1)}(\mathcal{W}+\nu)$, $\mathcal{W} \in \mathcal{B}(\mathbb{R}^n)$, we will have 
\begin{align}
    (Ce^{At})_{\#} \mu_0^{(2)}(\mathcal{V})=&\mu_0^{(2)}\left((Ce^{At})^{-1}(\mathcal{V})\right) \notag\\
    =&\mu_0^{(2)}\left((Ce^{At})^{-1}(\mathcal{V})-\nu \right) \label{shift}\\
    =&\mu_0^{(1)}\left((Ce^{At})^{-1}(\mathcal{V}) \right) \notag\\
    =& (Ce^{At})_{\#} \mu_0^{(1)}(\mathcal{V}) \notag,
\end{align}
for any $\mathcal{V} \in \mathcal{B}(\mathbb{R}^n)$ and $t \geq 0$, in which \eqref{shift} is due to the fact that $(Ce^{At})^{-1}(\mathcal{V})$ is closed under addition with scalar multiples of $\nu$, i.e., $(Ce^{At})^{-1}(\mathcal{V})=(Ce^{At})^{-1}(\mathcal{V})+c\nu$, $c \in \mathbb{R}$. This negates ensemble observability of $(A,C)$. We will state the result as a theorem for future references.

\begin{theorem}[Necessary condition for ensemble observability] \thlabel{necessary}
If an LTI ensemble characterized by structural dynamics $(A,C)$ is observable, then the pair $(A,C)$ is observable.
\end{theorem}

Note that \thref{necessary} is a result which holds for any LTI ensemble, including ones with discrete, absolutely continuous, or singular state distributions. However, the rest of this paper will focus on finding a sufficient condition for observability of LTI \textit{discrete} ensembles. To this end, we start with presenting some steps which will finally lead to the intended result, \thref{equivalent}. 

\begin{proposition} \thlabel{prop2}
An LTI discrete ensemble characterized by structural dynamics $(A,C)$ is observable if for any two initial state distributions $\mu_0^{(1)}, \mu_0^{(2)} \in \mathcal{P}_d(\mathbb{R}^{n})$, and every $x\in \mathbb{R}^n $, there exists $t\geq 0$ such that
\begin{align} \label{eq:40}
    \mu_0^{(1)}(\{x\}) < \mu_0^{(2)}(\{x\}) \Longrightarrow  (Ce^{At})_{\#} \mu_0^{(1)}\left(\{Ce^{At}x\}\right)=\mu_0^{(1)}(x).
\end{align}
\end{proposition}

\begin{proof}
Suppose the ensemble is not observable. Then by definition there exist initial state distributions $\mu_0^{(1)} \neq \mu_0^{(2)} \in \mathcal{P}_d(\mathbb{R}^{n})$ such that they are not distinguishable in the output at all, in the sense that $(Ce^{At})_{\#} \mu_0^{(1)} = (Ce^{At})_{\#} \mu_0^{(2)}$, for all $t \geq 0$. Then, if statement (3) holds, for every $x\in \mathbb{R}^n$ with $\mu_0^{(1)}(x) < \mu_0^{(2)}(x)$ we may correspondingly choose $T \geq 0$ so that $(Ce^{AT})_{\#} \mu_0^{(1)}\left(\{Ce^{AT}x\}\right)=\mu_0^{(1)}(\{x\})$. Then we will have
\begin{align}
    \mu_0^{(1)}(\{x\}) &< \mu_0^{(2)}(\{x\}) \\
                   &\leq (Ce^{AT})_{\#} \mu_0^{(2)}\left(\{Ce^{AT}x\}\right)\\
                   &=(Ce^{AT})_{\#} \mu_0^{(1)}\left(\{Ce^{AT}x\}\right)\\
                   &=\mu_0^{(1)}(\{x\}),
\end{align}
where inequality (5) comes from the fact that $L_{\#}\mu(\mathcal{V})=\mu(L^{-1}(\mathcal{V}))$, by definition. The derivation above shows a contradiction. Hence, in the case the ensemble $(A,C)$ is not observable, implication (3) is not valid either, and the proof is complete.
\end{proof}

\begin{remark}
Notice that the right hand side expression in the implication $\eqref{eq:40}$ does not involve $\mu_0^{(2)}$. Also,
\begin{align}
    (Ce^{At})_{\#} \mu_0^{(1)}\left(Ce^{At}x\right)=\mu_0^{(1)}(\{x\}) &\Longleftrightarrow \mu_0^{(1)}\left((Ce^{At})^{-1}(Ce^{At}x) \setminus \{x\}\right)=0 \\
    &\Longleftrightarrow \mu_0^{(1)}\left((x+\ker Ce^{At})\setminus \{x\}\right)=0, \label{eq:9}
\end{align}
in which the last expression does not involve $\mu_0^{(2)}$. So, the apparently stronger condition 
\begin{equation*}
     \forall \mu_0^{(1)} \in \mathcal{P}_d(\mathbb{R}^n), \hspace{5pt} \forall x \in \mathbb{R}^n, \hspace{5pt} \exists t \geq 0: \hspace{10pt} \mu_0^{(1)}\left((x+\ker Ce^{At})\setminus \{x\}\right)=0
\end{equation*}
is indeed equivalent to the condition $\eqref{eq:40}$, as it is vacuous for the case $\mu_0^{(1)}(\{x\})=1$, and if $\mu_0^{(1)}(\{x\})<1$, there surely exists a measure $\mu^{(2)}_0 \in \mathcal{P}_d(\mathbb{R}^n)$ such that $\mu_0^{(1)}(\{x\})<\mu_0^{(2)}(\{x\})$. So one may readily verify that the sufficient condition provided in \thref{prop2} can be simplified in the following way:
\end{remark} 

\begin{theorem}
An ensemble characterized by structural dynamics $(A,C)$ is observable if for every measure $\mu_0 \in \mathcal{P}_d(\mathbb{R}^n)$ and $x \in \mathbb{R}^n$, there exists $t \geq 0$ such that
\begin{equation} \label{eq:10}
    \mu_0\left((x+\ker Ce^{At}) \setminus \{x\} \right)=0.
\end{equation}
\end{theorem}

\begin{lemma} \thlabel{l:1}
For $\mu \in \mathcal{P}_d(\mathbb{R}^n)$, and $\mu-$measurable sets $\mathcal{A}$ and $\mathcal{B}$, we have
\begin{equation*}
    \mu(\mathcal{A} \setminus \mathcal{B})=0 \Longleftrightarrow \mathcal{A} \cap \mathcal{U} \subset \mathcal{B},
\end{equation*}
in which $\mathcal{U}$ is the set containing atoms of $\mu$.
\end{lemma}

\begin{proof}
\begin{align}
    \mu(\mathcal{A} \setminus \mathcal{B})=0 \Longleftrightarrow \mu\left((\mathcal{A} \cap \mathcal{U}) \setminus \mathcal{B}\right)=0 \xLeftrightarrow{((\mathcal{A} \cap \mathcal{U}) \setminus \mathcal{B}) \subset \mathcal{U}} &(\mathcal{A} \cap \mathcal{U}) \setminus \mathcal{B}=\varnothing \Longleftrightarrow \mathcal{A} \cap \mathcal{U} \subset \mathcal{B}.
\end{align}
\end{proof}

\begin{proposition} \thlabel{prop:5}
For every $\mu \in \mathcal{P}_d(\mathbb{R}^n)$ and $x \in \mathbb{R}^n$, there exists $t \geq 0$, such that $\mu\left((x+\ker Ce^{At}) \setminus \{x\} \right)=0$,
\underline{if and only if} for every countable set $\mathcal{I} \subset \mathbb{R}^n$, there exists $t \geq 0$ such that $\ker Ce^{At} \bigcap \mathcal{I} \subset \{0\}$.
\end{proposition}

\begin{proof}
($\Longrightarrow$) For each countable set $\mathcal{I}=\{x_1,x_2,...\} \subset \mathbb{R}^n$, we may associate a discrete measure $\mu$, for instance  $\Sigma_{n=1}^{\infty} \frac{1}{2^n} \delta_{x_n}$, whose atoms are members of $\mathcal{I}$. So, by hypothesis, $\mu\left(\ker Ce^{At}\setminus \{0\}\right)=0$, for some $t \geq 0$, which implies $\ker Ce^{At} \bigcap \mathcal{I} \subset \{0\}$, by \thref{l:1}.\\
($\Longleftarrow$) Fix $\mu \in \mathcal{P}_d(\mathbb{R}^n)$ and $x \in \mathbb{R}^n$, and let $\mathcal{U}=\{a_1,a_2,...\} \subset \mathbb{R}^n$ denote the set of atoms of $\mu$. Then, by hypothesis, we may find $t \geq 0$, such that $\ker Ce^{At} \bigcap \mathcal{V} \subset \{0\}$, where $\mathcal{V}:=\mathcal{U}-x=\{a_i-x| i=1,2,...\}$ is the translation of $\mathcal{U}$ by $-x$. Thus, again using \thref{l:1}, $\mu\left((x+\ker Ce^{At})\setminus \{x\}\right)=0$.
\end{proof}

Notice that, what \thref{prop:5} is providing is a condition equivalent to the condition $\eqref{eq:10}$, which itself is a sufficient condition for ensemble observability. So, now we have all the components to provide our first result on a sufficient condition for ensemble observability which resembles a similar result in \cite{Sampled Observability} for the case of \textit{finite} ensembles.

\begin{theorem}\thlabel{th:5}
A continuous-time, LTI, discrete ensemble characterized by structural dynamics $(A,C)$ is observable if for every countable set $\mathcal{I} \subset \mathbb{R}^n$, there exists $t \geq 0$ such that
\begin{equation} \label{eq:11}
    \ker Ce^{At} \bigcap \mathcal{I} \subset \{0\}.
\end{equation}
\end{theorem}

Given the result above, note that Proposition 5 in \cite{Sampled Observability} interestingly proves useful in the development of our results such that it guarantees that observability of $(A,C)$ is indeed a sufficient condition for $\eqref{eq:11}$. 

\begin{proposition}[Proposition 5 in \cite{Sampled Observability}] \thlabel{pr:6}
Let $(A,C)$ be an observable system. Then for every countable family $\mathcal{I}=\{x_1,x_2,...\} \subset \mathbb{R}^n$ of non-zero vectors, there exists $t \geq 0$ such that 
\begin{equation} \label{eq:12}
  \ker Ce^{At} \bigcap \mathcal{I}=\varnothing.
\end{equation}
\end{proposition}

Since \eqref{eq:11} and \eqref{eq:12} are clearly equivalent, what \thref{th:5} and \thref{pr:6} together with the fact that observability of $(A,C)$ is necessary for its ensemble observability (\thref{necessary}), imply is 
stated in the following theorem:

\begin{theorem} \thlabel{equivalent}
A continuous-time, LTI, discrete ensemble characterized by structural dynamics $(A,C)$ is observable in the sense of \eqref{ens_obs_def}, if and only if $(A,C)$ is observable.
\end{theorem}

\section{State Tracking of Linear Ensembles}

In studying the state tracking problem for linear ensemble $\eqref{lin_sys1}$, as described earlier in this chapter, we are supposed to obtain the state distributions $\{\hat{\mu_t}\}_{0\leq t \leq T}$ not only for the time instances $k=0,1,...,T$ at which the output distributions $\mu_k$ are observed, but for any $t \in [0,T]$. Among all possible $\hat{\mu}_t$, we are attempting to find those which are optimal with respect to some cost function in terms of the control input $u$ to the ensemble. Indeed, this attempt is motivated by the fact that in reality many examples of ensembles are observed to be evolving in some optimal fashion. So, if the cost function upon which $\{\hat{\mu_t}\}_{0\leq t \leq T}$ are optimal is quadratic in the control input $u$, then our setting for the state tracking problem fits within the framework of the stochastic control problem $\eqref{p1}$ in chapter 2. In this regard, we find it necessary to assume controllability of ensemble dynamics $(A(\cdot),B(\cdot))$ over the intervals $[k,k+1]$, $k=0,1,...,T-1$. Then, the task of obtaining the ensemble state distributions $\hat{\mu}_t$ translates into finding the distribution of the stochastic process $\mathbf{x}^u(t)$, $t \in [0, T]$, resulting from the following optimization
\begin{subequations} \label{eq:1000}
\begin{align}
    &\inf_{u \in \mathcal{U}} \mathbb{E}\left\{ \int_{0}^T \frac{1}{2}||u(t,\mathbf{x}^{u}(t))||^2\mathrm{d}t\right\} \\
    &\dot{\mathbf{x}}^u(t)=A(t)\mathbf{x}^u (t) + B(t)u(t,\mathbf{x}^{u}(t)) \\
    &\mathbf{y}(k)=C(k)\mathbf{x}^{u}(k)\sim \label{eq:4}\mu_k, \hspace{10pt} k=0,1,2,...,T, 
\end{align}
\end{subequations}
where $\mathcal{U}$, without loss of generality, is the set of feedback control laws $u: (t,x) \longmapsto u(t,x)$ satisfying $\eqref{eq:4}$. Next, following \cite{StateTracking} and motivated by the discussions leading to \thref{Th:1} in Chapter 2, we will have:

\begin{align}
    &\inf_{u \in \mathcal{U}} \mathbb{E}\left\{ \int_{0}^T \frac{1}{2}||u(t,\mathbf{x}^{u}(t))||^2\mathrm{d}t\right\}  \label{eq:5}\\
    &\hspace{20 pt} \dot{\mathbf{x}}^u(t)=A(t)\mathbf{x}^u (t) + B(t)u(t,\mathbf{x}^{u}(t)) \notag\\
    & \hspace{20 pt} \mathbf{y}(k)=C(k)\mathbf{x}^{u}(k)\sim \mu_k, \hspace{10pt} k=0,1,2,...,T, \notag\\
    =&\inf_{u \in \mathcal{U}}\sum_{k=0}^{T-1} \mathbb{E}\left\{ \int_{k}^{k+1} \frac{1}{2}||u(t,\mathbf{x}^{u}(t))||^2\mathrm{d}t\right\} \\
    &\hspace{20 pt} \dot{\mathbf{x}}^u(t)=A(t)\mathbf{x}^u (t) + B(t)u(t,\mathbf{x}^{u}(t)) \notag\\
    & \hspace{20 pt} \mathbf{y}(k)=C(k)\mathbf{x}^{u}(k)\sim \mu_k, \hspace{10pt} k=0,1,2,...,T,\notag \\
    =&\inf_{(u_0,u_1,...,u_{T-1}) \in \prod_{k=0}^{T-1}  \mathcal{U}_k}\sum_{k=0}^{T-1} \mathbb{E}\left\{ \int_{k}^{k+1} \frac{1}{2}||u_k(t,\mathbf{x}^{u}(t))||^2\mathrm{d}t\right\} \label{eq:6}\\
    &\hspace{20 pt} \dot{\mathbf{x}}^u(t)=A(t)\mathbf{x}^u (t) + B(t)u(t,\mathbf{x}^{u}(t)) \notag\\
    & \hspace{20 pt} \mathbf{y}(k)=C(k)\mathbf{x}^{u}(k)\sim \mu_k, \hspace{10pt} k=0,1,2,...,T,\notag \\
    = &\sum_{k=0}^{T-1} \inf_{u \in \mathcal{U}_k} \mathbb{E}\left\{ \int_{k}^{k+1} \frac{1}{2}||u(t,\mathbf{x}^{u}(t))||^2\mathrm{d}t\right\} \label{eq:7}\\
    &\hspace{20 pt} \dot{\mathbf{x}}^u(t)=A(t)\mathbf{x}^u (t) + B(t)u(t,\mathbf{x}^{u}(t)) \notag\\
    & \hspace{20 pt} \mathbf{y}(k)=C(k)\mathbf{x}^{u}(k)\sim \mu_k, \hspace{10pt} \mathbf{y}(k+1)=C(k+1)\mathbf{x}^{u}(k+1)\sim \mu_{k+1}, \notag\\
    =&\sum_{k=0}^{T-1} \inf_{T_{\#}\hat{\mu}_k=\hat{\mu}_{k+1}} \int_{\mathbb{R}^n} c_k\big(x,T(x)\big)\mathrm{d}\hat{\mu}_k(x)  \label{eq:8}\\
    &\hspace{20 pt} [C(k)]_{\#}\hat{\mu}_k=\mu_k, \hspace{10 pt}  [C(k+1)]_{\#}\hat{\mu}_{k+1}=\mu_{k+1}, \notag
\end{align}
where $c_k(x,y)=\frac{1}{2}\big(y-\Phi_{(k+1,k)} x\big)^{\top}W_{(k,k+1)}^{-1}\big(y-\Phi_{(k+1,k)} x\big)$, the set $\mathcal{U}_k$ in $\eqref{eq:6}$ is the set of restrictions $u\big|_{[k,k+1]\times \mathbb{R}^n}$ of the control laws $u \in \mathcal{U}$, and the equality $\eqref{eq:7}$ is by the fact that there is a one-to-one correspondence between $\mathcal{U}$ and $\prod_{k=0}^{T-1} \mathcal{U}_k$.

What the final expression $\eqref{eq:8}$ implies is that the state tracking problem with $T$ intermediate output distributions can be split into $T$ independent Monge optimal transport problems, each constrained with two "push-forward"-type equations. Now, it remains to provide a numerical algorithm to solve each of the optimal transport problems seeking $\hat{\mu}_k$, $k=0,1,...,T$, and then find the interpolation $\{\hat{\mu}_t\}_{k\leq t \leq k+1}$, out of $\hat{\mu}_k$ and $\hat{\mu}_k$ for each $k=0,1,2,...T-1$. Considering the equivalent Kantorovich representation of $\eqref{eq:8}$, one can easily observe that, any of those optimal transport problems, say the $k$-th, can be cast into a linear programming problem, using an ordering $\{z^{1},z^{2},...,z^{d^n}\}$ of the finite discretization $\{x^1,x^2,...,x^d\}^n \subset \mathbb{R}^n$ as the state space, in the following form:
\begin{subequations}
\begin{align}
    &\min_{\pi_k \in \mathbb{R}^{d^n \times d^n}} \sum_{i=1}^{d^n} \sum_{j=1}^{d^n} c^{(k)}_{ij}\pi_k(i,j), \\
    & \pi_k \mathbf{1}_d=\hat{m}_k,\\
    & \pi_k^{\top} \mathbf{1}_d=\hat{m}_{k+1},\\
    & \hat{m}_k(i)=\mu_k\left(C(k)z^i\right), \hspace{5pt} 1\leq i \leq d^n,
\end{align}
\end{subequations}
where $c^{(k)}_{ij}=c_k(z^{i},z^{j})=\frac{1}{2}\big(z^{j}-\Phi_{(k+1,k)} z^{i}\big)^{\top}W_{(k,k+1)}^{-1}\big(z^{j}-\Phi_{(k+1,k)} z^{i}\big)$, and $\mathbf{1}_d$ is the $d$-dimensional vector of all ones. It is worth noting that the discretization set $\{x^1,x^2,...,x^d\}^n$ in the $k-$th optimization problem should be chosen such that it conveys as much mass of the measure $\hat{\mu}_k$ as possible, and that at the same time, its dimension $d$ is not too large that the associated problem becomes computationally infeasible. Specifically, as an index of computational complexity, this numerical algorithm for the solution of the given state tracking problem requires solving for at least $Td^{(2n)}$ variables; a number which grows exponentially in the dimension of the structural dynamics, and linearly in the number of output distributions.

\subsection{State Tracking with Gaussian Marginals}

We observed earlier in this chapter that the state tracking problem \eqref{eq:1000} for linear ensembles could be transformed into a set of optimal transport problems seeking the marginals $\hat{\mu}_k$ as seen in \eqref{eq:8}. In this section we would like to study the special case that the ensemble output distributions $\mu_k$, observed at time instances $k=0,1,...,T$, are Gaussian $\mathcal{N}(\nu_k, \Sigma_k)$. First note that in a linear system with observable $(A(\cdot),C(\cdot))$, the initial state is uniquely determined by
\begin{equation} \label{eq:500}
\mathbf{x}(t_0)=M_{(t_0,t_1)}^{-1} \int_{t_0}^{t_1} \Phi_{(t,t_0)}^{\top} C(t)^{\top}\mathbf{y}(t)\mathrm{d}t
\end{equation}
where $M_{(t_0,t_1)}=\int_{t_0}^{t_1} \Phi_{(t,t_0)}^{\top} C(t)^{\top}C(t)\Phi_{(t,t_0)}\mathrm{d}t$ is the observability Gramian from $t_0$ to $t_1$. This implies that the initial states $\mathbf{x}(t_0)$ giving rise to a Gaussian process $\mathbf{y}(t)$ has to be Gaussian, since it is a linear function of the process $\mathbf{y}(t)$ (see e.g. \cite[Section 3.6]{Astrom}). In our setting of ensemble state tracking, the observed outputs distributions $\mu_k$ are Gaussian at each $k=0,1,2...,T$. But the property that the \textit{process} $\mathbf{y}(t)$, $0 \leq t \leq t$, be Gaussian, i.e., \textit{all} finite dimensional distributions of $\mathbf{y}(t)$ be multivariate Gaussian, is neither assumed, and nor required. In fact, in viewing \eqref{eq:500}, given the continuity of $C$ and the state transition matrix $\Phi$, what is required to conclude $\mathbf{x}(t_0)$ is Gaussian is that $\mathbf{y}(t)$ would have continuous sample paths, which is for free in linear systems with continuous coefficients. Now, the future states distributions are determined by $\hat{\mu}_t=(L_t)_{\#}\hat{\mu}_0$, where $L_t(\mathbf{x})=\Phi_{(t,0)}\mathbf{x}+\int_{0}^{t} \Phi_{(t,\tau)}B(\tau)u(\tau)\mathrm{d}\tau$ is a linear mapping. This readily implies that $\mu_t$, $0\leq t \leq T$ are Gaussian.

Letting $\hat{\mu}_t \longleftrightarrow (\hat{\nu}_t, \hat{\Sigma}_t)$, \cite{StateTracking} specifically proposes a two-stage method to initially find the means $\hat{\nu}_k$, $k=0,1,...,T$, and then the covariances $\hat{\Sigma}_k$ of the states by solving a semidefinite program (SDP). Subsequently, for each $k$, the set of interpolant pairs $(\hat{\nu}_t, \hat{\Sigma}_t)$, $k\leq t \leq k+1$, are obtained along with the solution of optimal transport between the distributions $\hat{\mu}_k \longleftrightarrow (\hat{\nu}_k, \hat{\Sigma}_k)$ and $\hat{\mu}_{k+1} \longleftrightarrow (\hat{\nu}_{k+1}, \hat{\Sigma}_{k+1})$, as described in \thref{Gaussian_remark} of chapter 2.
\begin{example}
Consider a linear ensemble with 
\begin{equation*}
    A= \begin{bmatrix} 0 &1\\0 & 0 \end{bmatrix}, \hspace{10pt} B= \begin{bmatrix} 0\\1 \end{bmatrix}, \hspace{10pt} C=\begin{bmatrix} 1 & 0 \end{bmatrix}
\end{equation*}
We will illustrate the solution of the ensemble state tracking problem for two cases using the aforementioned method: \\
\underline{case 1}:  scalar Gaussian outputs at $t=0,1,...,4$ with means and variances $[-1,3,5,-4,-7]$, and $[3, 3, 3, 3, 4]$.\\
\underline{case 2}:  scalar Gaussian outputs at $t=0,1,...,4$ with means and variances $[-1,3,5,-4,-7]$, and $[6, 7, 3, 2, 5]$.
\begin{figure}[H]
   \centering
   \includegraphics[width=15cm]{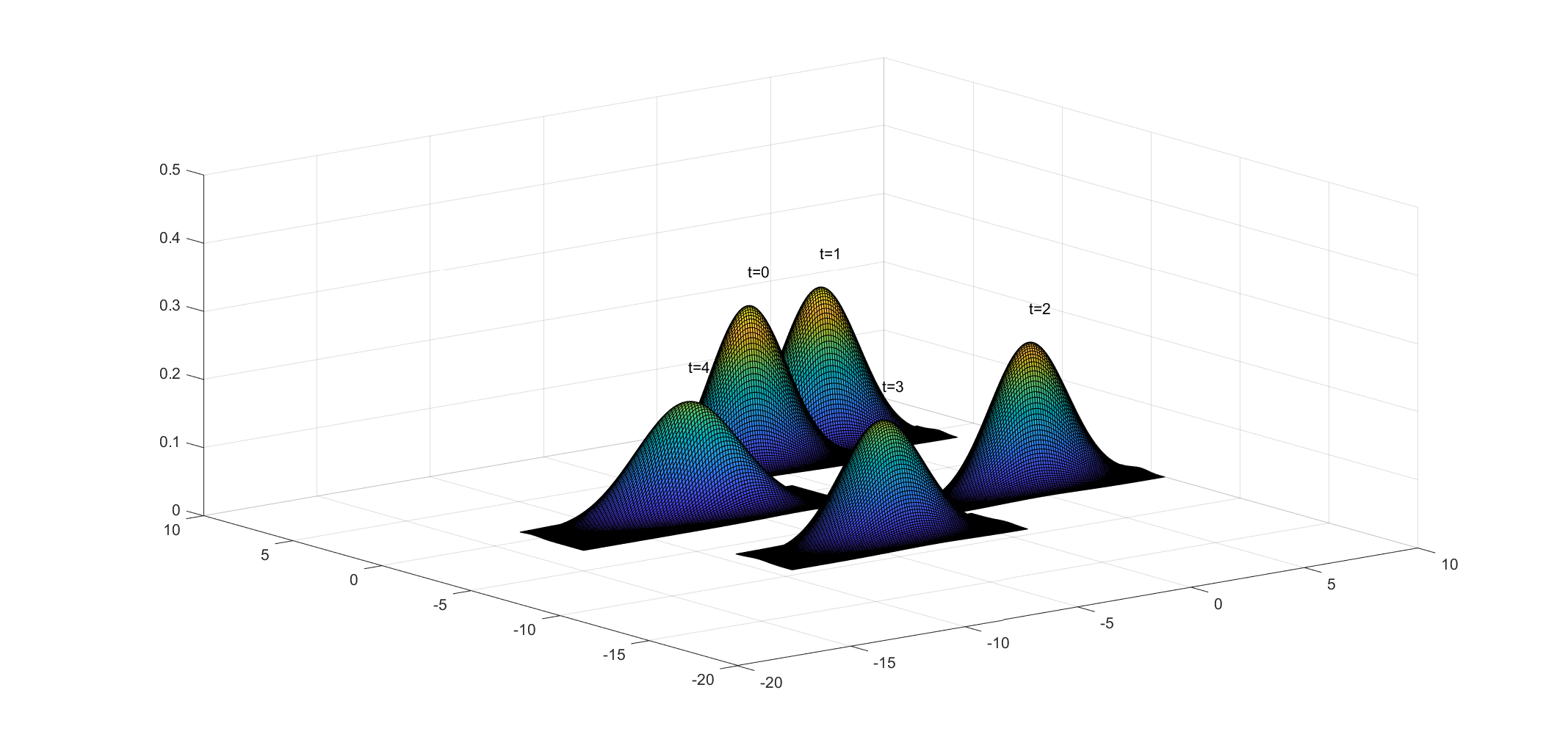}    \caption{ Tracking of state distributions $\hat{\mu}_t$, in case 1.}
\end{figure}
\begin{figure}[H]
    \centering
   \includegraphics[width=15cm]{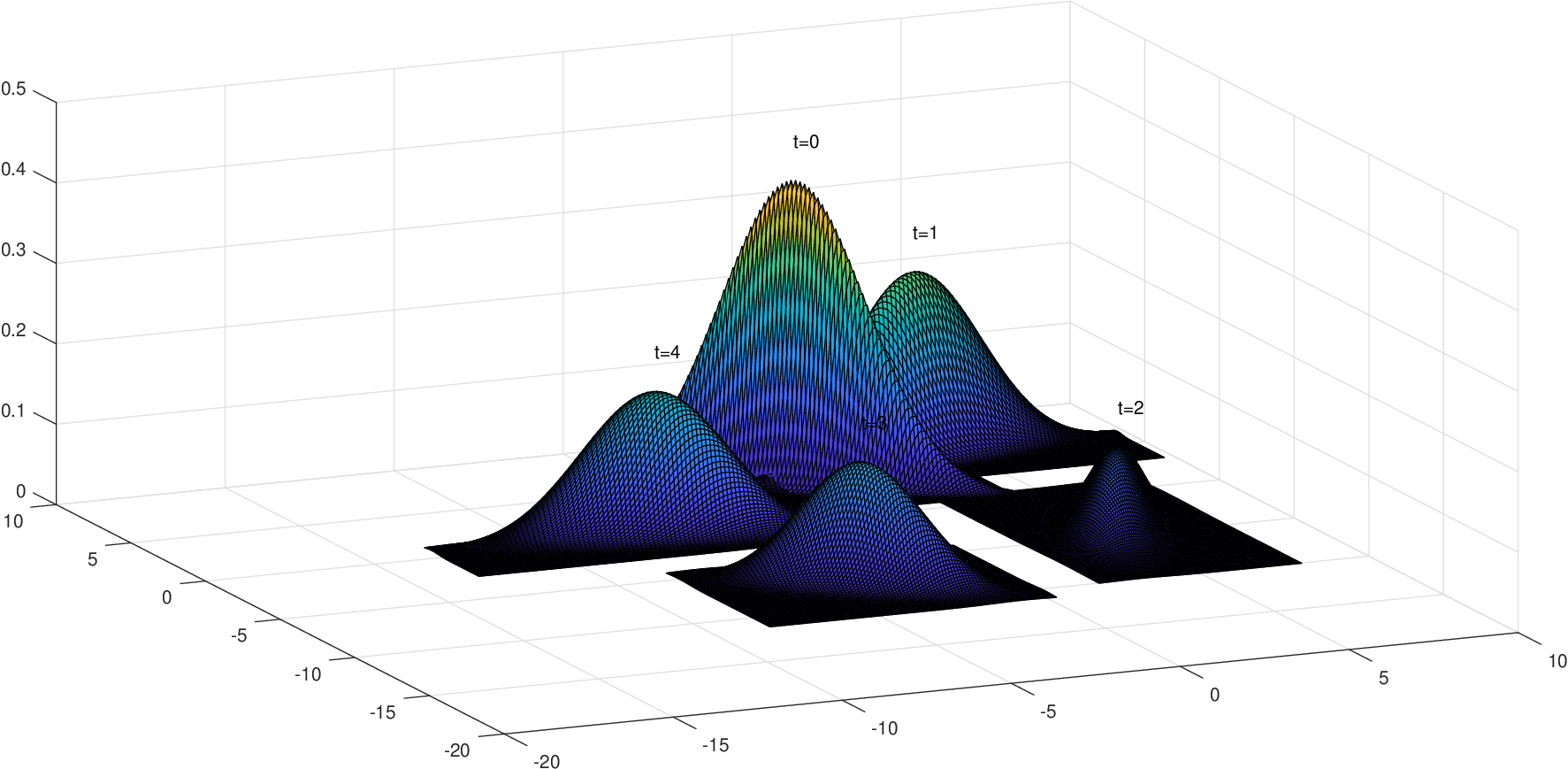}    \caption{ Tracking of state distributions $\hat{\mu}_t$, in case 2.}
    \label{fig:1}
\end{figure}

\end{example}

\clearpage

\section{Controllability of Ensembles}

Let us again consider an ensemble of systems with nonlinear dynamics

\begin{align} \label{gen_dyn2}
    \dot{\mathbf{x}}^{(i)}(t)&=f(t,\mathbf{x}^{(i)}(t),u(t)),\hspace{10pt}, \mathbf{x}^{(i)}(0)=\mathbf{x}^{(i)}_0 
\end{align}
indexed by $i \in I$, and note that, as in the previous sections, the state distributions at each time $t$ are characterized by a probability measure, denoted by $\mu_t \longleftrightarrow \{\mathbf{x}^{(i)}(t)\}_{i \in I}$. We start this section by providing some well-known definitions of ensemble controllability in the literature (see e.g. \cite{Approx&Exact, Li, xudong}) in order to illustrate width of the spectrum of ideas. Then, adopting one of those definitions as a reference, we will propose a continuous measure of controllabillity as compared with that reference binary controllability. 

\begin{definition}[Exact Controllability] \thlabel{def_exact}
The ensemble \eqref{gen_dyn2} is said to be exactly controllable from $\mu_0$ to $\mu_1$ on $[0,T]$, $T>0$, if for each $\epsilon>0$ there exists a (piecewise) continuous control input $u$ such that the displacement interpolation $\mu(\cdot)$ of the ensemble satisfies $\mu(T)=\mu_1$.
\end{definition}

\begin{definition}[Approximate Controllability]
The ensemble \eqref{gen_dyn2} is said to be approximately controllable from $\mu_0$ to $\mu_1$ on $[0,T]$, $T>0$, if for each $\epsilon>0$ there exists a (piecewise) continuous control input $u$ such that the displacement interpolation $\mu(\cdot)$ of the ensemble satisfies
\begin{equation}
    W_p(\mu(T),\mu_1) < \epsilon.
\end{equation}
where $W_p$ is the $p$-Wasserstein distance.
\end{definition}

Another commonly used definition for controllability of ensemble is the so called $\mathbb{L}_p$-controllability which will come next. This definition, in fact, assumes a setting for studying controllability of ensembles a little different from that for other definitions of ensemble controllability and even for the observability of ensembles as defined in this paper. And that is, in this setting, of course the control inputs are uniformly influencing all systems (i.e., broadcast control), at each time $t$ the same state distributions with different configurations with respect to the index set $I$ are distinguished. In other words, by controllability of the ensemble in the following sense, we require that not only a final distribution $\mu_1 \longleftrightarrow \{\mathbf{x}^{(i)}(T)\}_{i \in I}$ be reachable from an initial distribution $\mu_0 \longleftrightarrow \{\mathbf{x}^{(i)}(0)\}_{i \in I}$, but also a function $\mathbf{x}(T):i \longmapsto \mathbf{x}^{(i)}(T)$ be reached from another function $\mathbf{x}(0):i \longmapsto \mathbf{x}^{(i)}(0)$ representing the configurations of the target and initial states, respectively. For each $t \geq 0$, a function $\mathbf{x}(t): I \rightarrow \mathbb{R}^n$ is called the \textit{ensemble profile} at time $t$.

\begin{definition}[$\mathbb{L}_p$-controllability]
The ensemble \eqref{gen_dyn2} is said to be $\mathbb{L}_p$-controllable, $1\leq p \leq \infty$, if for any initial and target profiles $\mathbf{x}_0,\mathbf{x}_d \in \mathbb{L}_p(I, \mathbb{R}^n)$, and $\epsilon>0$, there exists a finite $T>0$ and a (piecewise) continuous control input $u$ such that the profile $\mathbf{x}(T)$ of the ensemble at time $T$ satisfies
\begin{equation}
    \left\Vert \mathbf{x}_d^{(\theta)}-\mathbf{x}^{(\theta)}(T)\right \Vert_{L_p}<\epsilon,
\end{equation}
where for each $1 \leq p <\infty$ and $x \in \mathbb{L}_p(I, \mathbb{R}^n)$ we have
\begin{equation}
    \left\Vert x \right\Vert_{\mathbb{L}_p}=\left(\int_{\theta \in I} \left\Vert x^{(\theta)}\right\Vert ^p d\theta\right)^{1/p}
\end{equation}
and for $p=\infty$, $||x||_{\mathbb{L}_{\infty}}=\text{ess\,sup}_{\theta} \left\Vert x^{(\theta)}\right\Vert$.
\end{definition}
\subsection{A Controllability Measure for Nonlinear Ensembles}

In this section, we shall consider an ensemble of particles indexed by $i \in I$, under the nonlinear dynamics
\begin{equation} \label{eq:34}
\dot{\mathbf{x}}^{(i)}(t)=f\Big(\mathbf{x}^{(i)}(t),\mathds{1}_D\left(\mathbf{x}^{(i)}(t)\right) u\left(t,\mathbf{x}^{(i)}(t)\right)\Big), \hspace{10pt} \mathbf{x}^{(i)}(0)=\mathbf{x}^{(i)}_0,
\end{equation}
over $0\leq t \leq T$, where $D \subset \mathbb{R}^n$ is a nonempty open connected set, $\mathds{1}_D$ is the indicator function of $D$, and $f: \mathbb{R}^n \times \mathbb{R}^m \rightarrow \mathbb{R}^n$ is continuous. 
Motivated by crowds and flocks models (e.g. \cite{CuckerSmale}, \cite{Control to flock}, \cite{pedestrian}) we may think of this dynamics as the dynamics of an ensemble whose states (agents positions, for instance) are influenced by the velocity vector field \begin{equation}
    V(x)=f(x,0),
\end{equation}
perturbed by the control input $(t,x) \mapsto u(t,x)$ which is influencing only a fixed portion $D$ of the space (see Figure \ref{fig:2}). Thus, if $\rho(t,x)$ is the time-varying density, with respect to the Lebesgue measure, of spatial distribution of the ensemble states, it is the (weak) solution of the corresponding continuity equation
\begin{equation} \label{eq:35}
    \frac{\partial \rho}{\partial t}(t,x)+\nabla_{x} \cdot \big[\rho(t,x)f\big(x,\mathds{1}_D(x) u(t,x)\big)\big]=0, \hspace{10pt} \forall (t,x),
\end{equation}
with $\rho(0,\cdot)=\rho_0$, where $\rho_0$ represents the initial distribution of the states. For more detailed discussions on the continuity equation, we may refer to \thref{r:1} in chapter 2, and the references therein. Note that, in the related literature, the transport-type equation \eqref{eq:35} is usually considered as the (macroscopic) alternative for the (microscopic) representation \eqref{eq:34} of the ensemble. So, in this chapter, we may interchangeably use either representations as the ensemble dynamics.\\

\begin{figure}[htp]
    \centering
    \includegraphics[width=10cm]{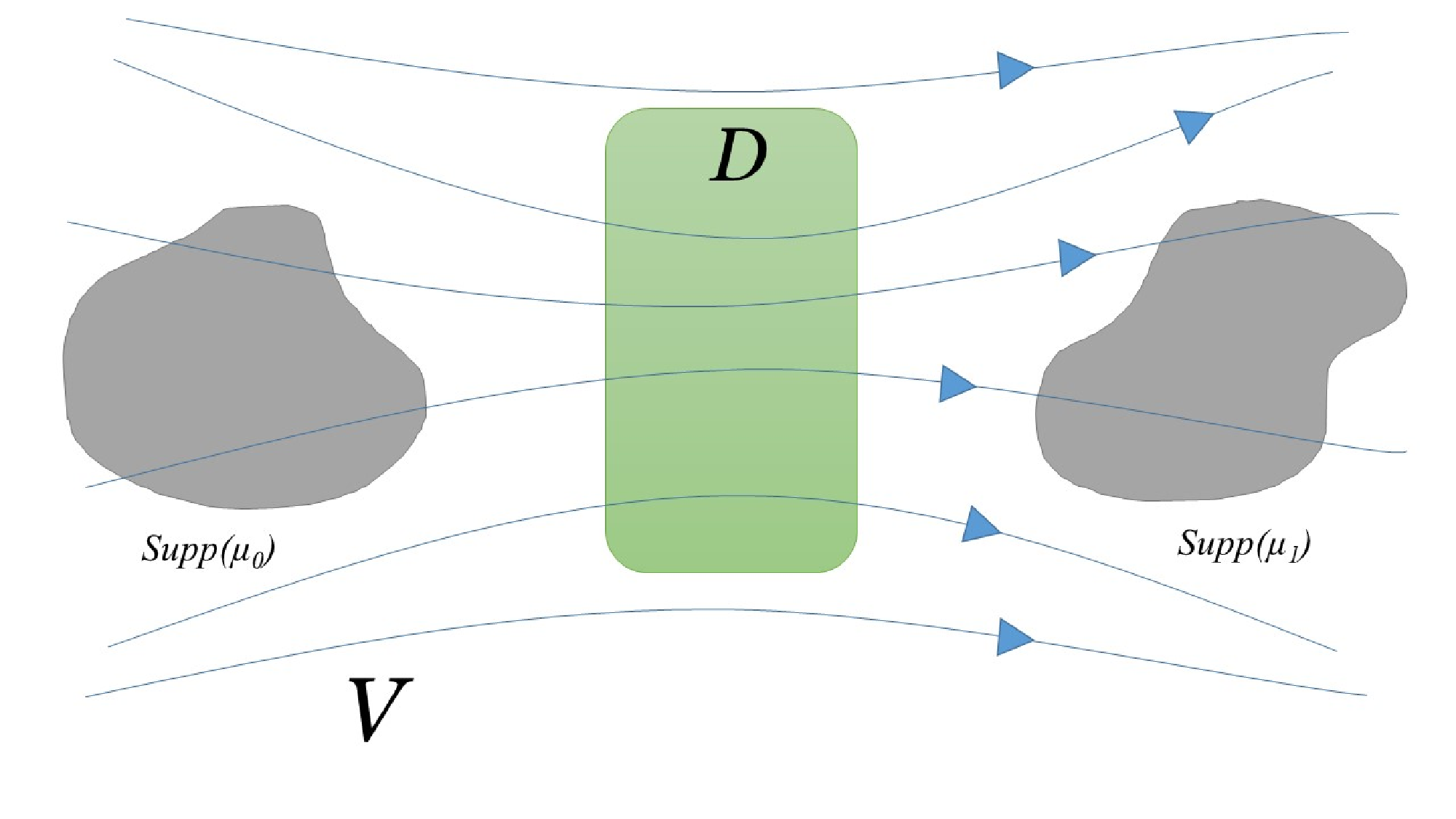}
    \caption{ Geometry of the region $D$, supports of $\mu_0$ and $\mu_1$, and the field $V$.}
    \label{fig:2}
\end{figure}
Our goal in this section is to define a controllability measure which is invariant under linear transformations of the ensemble states, specially due to the physical units scalings. In this regard, consider a finite ensemble of $N$ agents with initial states $\mathbf{x}_0=\{\mathbf{x}_0^1,\mathbf{x}_0^2,...,\mathbf{x}_0^N\}$, and target states $\mathbf{x}_1=\{\mathbf{x}_1^1,\mathbf{x}_1^2,...,\mathbf{x}_1^N\}$. Correspondingly, we may attribute the distributions
\begin{equation} \label{empirical}
    \mu_0=\frac{1}{N}\sum_{i=1}^N \delta_{\mathbf{x}_0^i}, \hspace{10pt}  \mu_1=\frac{1}{N}\sum_{i=1}^N \delta_{\mathbf{x}_1^i},
\end{equation}
to the initial and target distributions of the ensemble states. Then, assuming exact controllability (\thref{def_exact}), let us define 
\begin{equation} 
    M(\mathbf{x}_0,\mathbf{x}_1):=\max_{i=1,2,...,N} |t_i^0+t_i^1|
\end{equation}
as the controllability measure for the ensemble between $\mathbf{x}_0$ and $\mathbf{x}_1$, in which $t_i^0=t^0(\mathbf{x}_i^0)$, and $t_i^1=t^1(\mathbf{x}_i^1)$, where
\begin{subequations}\label{times}
\begin{align} 
    t^0(x):=&\inf\{t \in \mathbb{R}_{+}: \phi^V_t(x) \in D\},\\
    t^1(x):=&\inf\{t \in \mathbb{R}_{+}: \phi^V_{-t}(x) \in D\},
\end{align}
\end{subequations}

and, $t \longmapsto \phi^V_t(x)$ and $t \longmapsto \phi^V_{-t}(x)$ are the forward and backward flows passing through $x$ and induced by the uncontrolled vector field $V$. But the struggle, here, is that, in order to compute the measure $M(\mathbf{x}_0,\mathbf{x}_1)$ we need to evaluate the minimum times $t_i^0=t^0(\mathbf{x}_i^0)$, and $t_i^1=t^1(\mathbf{x}_i^1)$, $i=1,2,...,N$. And to evaluate $t_i^0$, for instance, we need to measure how long it takes for the $i$-th agent with initial state $\mathbf{x}_i^0 \neq 0$ to get to the control region $D$ under the uncontrolled field $V$. However, this necessitates setting up the ensemble at the initial configuration $\{0,0,...,0,\mathbf{x}_i^0,0,...,0\}$, and letting it run through $V$, which is not practical in many circumstances under the ensemble setting. So, as an example, the following proposed result (\cite{MinimalTime}), leads us to algorithms which are computationally feasible in practice.
\begin{theorem}
Given a finite ensemble with dynamics \eqref{eq:34}, and the initial and final states $\mathbf{x}_0$ and $\mathbf{x}_1$ charachterized by the measures $\eqref{empirical}$, the controllabaility measure $M(\mathbf{x}_0,\mathbf{x}_1)$ is equal to 
$$S(\mu_0,\mu_1):=\sup_{m \in [0,1]} \{\mathcal{F}_{\mu_0}^{-1}(m)+\mathcal{H}_{\mu_1}^{-1}(1-m)\},$$ where 
\begin{align}
    \mathcal{F}_{\mu_0}^{-1}(m)&=\inf \{t \in \mathbb{R}_{+}: \mathcal{F}_{\mu_0}(t)\geq m\},\\
    \mathcal{H}_{\mu_1}^{-1}(m)&=\inf \{t \in \mathbb{R}_{+}: \mathcal{H}_{\mu_1}(t)\geq m\},
\end{align}
in which 
\begin{align}
    \mathcal{F}_{\mu_0}(t)&=\mu_0(\{x \in supp(\mu_0): t^0(x) \leq t\},\\
    \mathcal{H}_{\mu_1}(t)&=\mu_1(\{x \in supp(\mu_1): t^1(x) \leq t\},
\end{align}
are the masses of the distributions $\mu_0$ and $\mu_1$, respectively, that reach, to the control region $D$ from $supp(\mu_0)$, and to $supp(\mu_1)$ from the region $D$, respectively, under the uncontrolled field $V$, in \underline{at most $t$ time units}.  
\end{theorem}

Note that the result above requires some mild geometric and analytic conditions on the distributions $\mu_0$, $\mu_1$, the vector field $V$ and the control input $u$. We take those conditions satisfied for simplicity now. Also, following the result above, in order to see an example of the structure of algorithms for computation of $S(\mu_0,\mu_1)$ , for the space dimension $n=2$ , we refer to \cite[Section 5]{MinimalTime}.\\
\begin{remark}[Invariance of the Controllability Measure]
Let us suppose the states $\mathbf{x}_0$, $\mathbf{x}_1$, of the finite ensemble 
above are linearly transformed by some matrix $P$. In this case, although the distributional characterization of the states will change to 
\begin{equation} \label{empirical2}
    \Tilde{\mu}_0=\frac{1}{N}\sum_{i=1}^N \delta_{P\mathbf{x}_0^i}, \hspace{10pt}  \Tilde{\mu}_1=\frac{1}{N}\sum_{i=1}^N \delta_{P\mathbf{x}_1^i},
\end{equation}
the control region $D$ and the velocity vector field $V$ will transform to $\Tilde{D}=P(D)$, and $\Tilde{V}(x)=Pf(P^{-1}x,0)$, correspondingly. Hence the minimum times $\eqref{times}$ will obviously remain unchanged, which implies \begin{equation}
    M_{\Tilde{D}}(P\mathbf{x}_0,P\mathbf{x}_1)=M_D(\mathbf{x}_0,\mathbf{x}_0).
\end{equation}.
\end{remark}

\renewcommand{\thechapter}{4}

\chapter{Conclusions and Recommendations}

In the current paper, we studied controllability and observability as two prominent systems theoretic concepts were studied in the context of ensembles. In the observability part, the problem of state tracking for ensembles with linear state and output dynamics was focused on in detail with a viewpoint from optimal transport theory. In this regard, the problem was first transformed into a linear quadratic problem with stochastic endpoints, analytic solution in terms of optimal transport map was identified, and a numerical algorithm to achieve it was proposed. In the controllability part, besides addressing different definitions of ensemble controllability, a measure of controllability for nonlinear ensembles was proposed with the  desired property of invariance under linear state transformations.

For the future studies, most importantly, generalization of the state tracking problem for nonlinear and discrete-time ensembles is recommended. Furthermore, in this paper, the admissible control inputs were restricted to deterministic feedback laws. So, the next natural step could be investigating the class of stochastic feedback laws with the insights from optimal transport theory.

\titleformat{\chapter}
{\normalfont\large}{Appendix \thechapter:}{1em}{}

\renewcommand{\baselinestretch}{1}

\bibliographystyle{unsrt}



\end{document}